\documentclass[12pt,a4paper]{amsart}
\usepackage[OT1]{fontenc}
\usepackage[utf8]{inputenc}
\usepackage[sc]{mathpazo}
\usepackage{amsmath,amssymb,amsfonts,mathrsfs}
\usepackage{mathabx}
\usepackage{graphicx}
\usepackage{soul}
\usepackage{pdfpages}
\usepackage[parfill]{parskip}
\usepackage{todonotes}
\usepackage{enumerate}

\numberwithin{equation}{section}
\usepackage[margin=2.9cm]{geometry}
\usepackage{epstopdf}

\numberwithin{equation}{section}

\newtheorem{theorem}{Theorem}[section]

\newtheorem{remark}[theorem]{Remark}
\newtheorem{corollary}[theorem]{Corollary}

\newtheorem{lemma}[theorem]{Lemma}
\newtheorem{proposition}[theorem]{Proposition}
\newtheorem{assumption}[theorem]{Assumption}

\theoremstyle{nonumberplain}

\newcommand{\RR}{\mathbb{R}}

\newcommand{\EE}{\mathbb{E}}

\newcommand{\E}[1]{\mathbb{E}\left[#1\right]}                     % expectation
\newcommand{\set}[1]{\left\{#1\right\}}
\newcommand{\Rplus}{\mathbb{R}_{\geq 0 }}

\renewcommand{\epsilon}{\ensuremath\varepsilon}

\renewcommand{\phi}{\ensuremath{\varphi}}

\begin{document}

\title[A comparison principle between rough and non-rough Heston models]{A comparison principle between rough and non-rough Heston models -- with applications to the volatility surface}

\author[M. Keller-Ressel]{Martin Keller-Ressel}
\address{TU Dresden, Institute of Mathematical Stochastics}
\email{martin.keller-ressel@tu-dresden.de}
\author[A. Majid]{Assad Majid}
\address{TU Dresden, Institute of Mathematical Stochastics}
\email{assad.majid@tu-dresden.de}
\begin{abstract}We present a number of related comparison results, which allow to compare moment explosion times, moment generating functions and  critical moments between rough and non-rough Heston models of stochastic volatility. All results are based on a comparison principle for certain non-linear Volterra integral equations. Our upper bound for the moment explosion time is different from the bound introduced by Gerhold, Gerstenecker and Pinter (2018) and tighter for typical parameter values. The results can be directly transferred to a comparison principle for the asymptotic slope of implied volatility between rough and non-rough Heston models. This principle shows that the ratio of implied volatility slopes in the rough vs.\ the non-rough Heston model increases at least with power-law behavior for small maturities. 
\end{abstract}
\maketitle
\section{Introduction} 
It is well-known that classic stochastic volatility models are not able to accurately reproduce all features of the observed implied volatility surface. In particular for short maturities, it is frequently seen that Markovian diffusion-driven models, such as the Heston model \cite{heston1993closed}, produce a smile which is flatter and less skewed  than the implied smile of observed market data \cite{bakshi1997empirical}. While adding jumps to the stock price dynamics can mitigate some of these deficiencies (cf. \cite{bates1996jumps, jacquier2013large}) a recently emerging alternative is given by \textit{rough volatility models} \cite{gatheral2018volatility}. In such models, volatility is modeled by a non-Markovian stochastic process comparable to fractional Brownian motion with low Hurst index (e.g. $H \approx 0.1$). Apart from the more realistic behavior of implied volatility, the rough volatility approach is also supported by econometric analyses of volatility time series \cite{gatheral2018volatility, fukasawa2019volatility}. While the behavior of implied volatility in rough models has mainly been studied by numerical computation, analytic results have been obtained e.g. in \cite{forde2017asymptotics, guennoun2018asymptotic, forde2019small} (short- and/or long-maturity asymptotics) and \cite{fukasawa2017short, bayer2019short} (short-time asymptotics of at-the-money skew).

Here, we focus on wing asymptotics (small- and large-strike) of implied volatility in the rough Heston model of \cite{eleuch2019characteristic} (see also \cite{eleuch2018perfect}), which is becoming increasingly popular due to its tractability and its connections with affine processes (see \cite{jaber2017affine, gatheral2018affine, keller-ressel2018affine}).\\
Starting with the results of \cite{lee2004moment} (see also \cite{benaim2009black}) it has become well-understood that wing asymptotics of implied volatility are intimately connected to moment explosions in the underlying stochastic model (see also \cite{friz2010encyclopedia}). Consequently, a first study of moment explosions in the rough Heston model has been undertaken by Gerhold, Gerstenecker and Pinter \cite{gerhold2018moment}. The authors derive a lower and upper bound for moment explosion times and a method for their numerical approximation (valid in a certain parameter range). We build on the results of \cite{gerhold2018moment} and derive a new upper bound for the moment explosion time in the rough Heston model (Thm.~\ref{thm:heston_upper_bound}). Our new bound is usually tighter than the upper bound of \cite{gerhold2018moment} and, more importantly, given by a transformation of the classic Heston explosion time, thus allowing for direct comparison between rough and non-rough Heston models. The result rests on a comparison principle for non-linear Volterra integral equations and leads to a number of further comparison results: For the moment generating functions of rough and non-rough Heston model (Thm.~\ref{thm:mgf_comparison}), for their critical moments (Thm.~\ref{thm:crit_mom_bound}) and finally for the implied volatility slope in the wings of the smile. We highlight Theorem~\ref{thm:implied_vola_bound}, which concerns the asymptotic slope of left-wing implied volatility ($AIVS_\alpha^-(T)$) in a (negative-leverage) rough Heston model in dependency on maturity $T$ and the roughness parameter $\alpha = H + \tfrac{1}{2}$. This slope can be lower bounded by the time-changed and rescaled slope of a non-rough Heston model ($AIVS_1^-(T)$) as
	\begin{equation*}
		AIVS_{\alpha}^- (T) \ge \frac{T^{\alpha -1}}{\alpha \Gamma(\alpha )} AIVS_1^-\left(\frac{T^{\alpha}}{\alpha \Gamma(\alpha)}\right)
	\end{equation*}
	for all $T$ smaller than a certain threshold $\mathfrak{T}'_\alpha$. 
Slightly weaker results that also apply to the right wing are finally given in Theorem~\ref{thm:AIVS_smallT}.

\section{Preliminaries}
\subsection{The rough Heston model}

We consider the rough Heston model \cite{eleuch2019characteristic, eleuch2018perfect} for a risk-neutral asset-price $S$ with spot variance $V$, given by
\begin{subequations}\label{eq:rough_Heston}
\begin{align}
dS_t &= S_t \sqrt{V_t} dW_t\\
V_t &= V_0 + \int_0^t \kappa_\alpha(t-s) \lambda(\theta(s) - V_s)ds + \eta \int_0^t \kappa_\alpha(t-s) \sqrt{V_s} dB_s
\end{align}
\end{subequations}
where $V_0$, $\lambda$ and $\eta$ are positive,  $\theta \in L^1_\text{loc}(\Rplus, \Rplus)$, $(B, W)$ are Brownian motions with constant correlation $\rho \in (-1,1)$ and $\kappa_\alpha$ is the power-law kernel
\[\kappa_\alpha(t) = \frac{1}{\Gamma(\alpha)} t^{\alpha-1}, \qquad \alpha \in (\tfrac{1}{2},1].\]
It was shown in \cite[Thm.~2.1]{eleuch2019characteristic} that $V$ is H\"older-continuous with exponent in $[0,\alpha - \tfrac{1}{2})$ and therefore that $\alpha$ controls the `roughness' of the variance process $V$. Other important properties of the rough Heston model, such as the decay of at-the-money implied volatility slope are also linked to the parameter $\alpha$, cf. \cite[Sec.~5.2]{eleuch2019characteristic}.\\
In the particular case $\alpha = 1$, the kernel becomes constant, i.e., $\kappa_1 \equiv 1$ and $V$ can be written in the familiar SDE form
\begin{equation}\label{eq:nonrough_Heston}
dV_t = \lambda (\theta(t) - V_t) dt + \eta \sqrt{V_t} dB_t.
\end{equation}
In this case $(S,V)$ becomes an extended Heston model with time-varying mean reversion level, as considered by \cite[Ex.~3.4]{buehler2006volatility} in the context of variance curve models. If also $\theta$ is constant, the Heston model of \cite{heston1993closed} (`classic Heston model') is recovered.

\subsection{The moment generating function of the rough Heston model}

Our comparison principle rests on the moment generating function of $X = \log S$, which has been studied in \cite{eleuch2018perfect} and \cite{gerhold2018moment}.\footnote{Related results for more general kernels $\kappa$ and multivariate generalizations (`affine Volterra processes') can be found in \cite{jaber2017affine, gatheral2018affine}} Define the \emph{moment explosion time} of the $\alpha$-rough Heston model
\begin{equation}\label{eq:moment_explosion}
T_\alpha^*(u) := \sup \set{t \ge 0: \E{e^{uX_t}} < \infty}, \qquad u \in \RR
\end{equation}
and the quadratic polynomial
\begin{equation}\label{eq:Rdef}
R(u,w) := \tfrac{1}{2}u(u-1) + w (\rho \eta u - \lambda) +\frac{\eta^2}{2}w^2,
\end{equation}
which can be considered the `symbol' of the Heston model in the sense of \cite{hoh1998symbolic}. Moreover denote the Riemann-Liouville left-sided fractional integral and derivative operators by
\begin{equation*}
	I^{\alpha} f(t) := \frac{1}{\Gamma(\alpha)} \int_0^t (t-s)^{\alpha - 1} f(s) \, ds \quad \text{and} \quad D^{\alpha} f(t) := \frac{d}{d t} I_t^{1-\alpha} f(t)
\end{equation*}
The moment generating function of the rough Heston model is given by the following result:
\begin{theorem}[\cite{eleuch2018perfect, gerhold2018moment}]\label{thm:mgf1}
In the rough Heston model, the log-price $X = \log S$ satisfies 
\begin{equation}\label{eq:rough_Heston_char_func}
	\EE \left[ e^{u X_t} \right] = \exp \left( \lambda \int_0^t \theta(t-s) \psi_{\alpha}(s,u)ds  + V_0 I^{1-\alpha}_t \psi_{\alpha}(t,u) \right)\end{equation}
for all $u \in  \RR$, $t \in [0,T_\alpha^*(u))$ and $\psi_{\alpha}(\cdot,u)$ solves the fractional Riccati equation
\begin{equation}\label{eq:frac_Riccati}
	 D^{\alpha} \psi_{\alpha}(t,u) = R(u,\psi_{\alpha}(t,u)), \quad I^{1-\alpha} \psi_{\alpha}(0,u) = 0.
\end{equation}
\end{theorem}
In the case of the non-rough Heston model (i.e, $\alpha = 1$) the operator $I^{1-\alpha}$ vanishes, $D^{\alpha}$ becomes an ordinary derivative and the fractional Riccati equation \eqref{eq:frac_Riccati} turns into the familiar Riccati ODE. Moreover, the solution $\psi_1$ and the moment explosion time $T^*_1(u)$ are explicitly known (cf. \cite{andersen2007moment, keller-ressel2011moment}) in the classic Heston case. The above theorem is complemented by the following result:
\begin{theorem}[\cite{gerhold2018moment}]\label{thm:mgf2}
The fractional Riccati equation \eqref{eq:frac_Riccati} is equivalent to the Riccati-Volterra integral equation
\begin{equation}\label{eq:Volterra_Riccati}
\psi_\alpha(t,u) = \int_0^t \kappa_\alpha(t-s) R(u,\psi_\alpha(s,u))ds
\end{equation}
and $T^*_\alpha(u) = \hat T_\alpha(u)$, where
\begin{equation}\label{eq:Tstar}
\hat T_\alpha(u) := \sup \set{t \ge 0: \psi_\alpha(t,u) < \infty}.
\end{equation}
\end{theorem}
For the equivalence of \eqref{eq:frac_Riccati} and \eqref{eq:Volterra_Riccati} see also \cite[Thm.~3.10]{kilbas2006theory}. The second part of the theorem states that the functions $t \mapsto \E{e^{uX_t}}$ and $t \mapsto \psi_\alpha(t,u)$ blow up at exactly the same time. Therefore, the solution $\psi_\alpha$ of \eqref{eq:Volterra_Riccati} contains all relevant information needed for the analysis of both moments and moment explosions in the rough Heston model.

\subsection{Calibration to the forward variance curve}\label{sec:calibration}
Given a stochastic volatility model with spot variance process $V$, the associated \emph{forward variance curve} is given by
\begin{equation}
\xi(T) := \E{V_T},
\end{equation}
and represents the market expectation of future variance. It is well-understood that forward variance is closely linked to the prices of variance swaps and other volatility-dependent products and that for these products the forward variance curve has a similar role as the forward curve for interest rates. In the rough Heston model \eqref{eq:rough_Heston}, it is known from \cite[Prop.~3.1]{eleuch2018perfect} (see also \cite{keller-ressel2018affine}) that the forward variance curve is given by 
\begin{equation}\label{eq:fw_rough_Heston}
\xi(T) = V_0 \left(1 - \int_0^T r_{\alpha,\lambda}(s) ds\right) + \int_0^T \theta(T-s) r_{\alpha,\lambda}(s) ds,
\end{equation}
where $r_{\lambda,\alpha}$  is the so-called resolvent of $\lambda \kappa_\alpha$, given by
\begin{equation}\label{eq:resolvent}
r_{\lambda,\alpha}(t) = \begin{cases} \lambda t^{\alpha-1} E_{\alpha,\alpha}(-\lambda t^\alpha), &\qquad \alpha \in (0,1)\\ \lambda e^{-\lambda t} &\qquad \alpha=1\end{cases}
\end{equation}
with $E_{\alpha,\alpha}$ denoting the Mittag-Leffler function, cf. \cite{haubold2011mittag}. Given a variance curve $\xi$ of suitable regularity, equation~\eqref{eq:fw_rough_Heston} can be inverted and solved for $\theta$, with solution
\begin{equation}\label{eq:theta_calibrated}
\theta(t) = \tfrac{1}{\lambda} D^\alpha \left(\xi(t) - V_0\right) + \xi(t),
\end{equation}
see also \cite[Rem.~3.2]{eleuch2018perfect}. We refer to a model with this choice of $\theta(.)$ as \textit{calibrated} to a given forward variance curve. For a calibrated model, the moment generating function \eqref{eq:rough_Heston_char_func} can be expressed in terms of the variance curve $\xi(.)$ instead of $\theta(.)$ and written as
\begin{equation}\label{eq:mgf_alternative}
	\EE \left[ e^{u X_t} \right] = \exp \left(\int_0^t \xi(t-s) \left(R(u,\psi_\alpha(s,u)) + \lambda \psi_\alpha(s,u) \right) ds \right),
\end{equation}
for all $u \in \RR$, $t \in [0,T_*(u))$, see \cite{keller-ressel2018affine}. 

\section{A comparison principle for Riccati-Volterra equations}

Our comparison results for moments, moment explosion times and implied volatilities in the rough Heston model will all be derived from comparison results for the Volterra-Riccati integral equation \eqref{eq:Volterra_Riccati}. In fact, the comparison results in this section are obtained for the more general Volterra integral equation
\begin{equation}\label{eq:Volterra_Riccati_general}
\psi_\kappa(t,u) = \int_0^t \kappa(t-s) R(u,\psi_\kappa(s,u))ds,
\end{equation}
where only the following assumptions on the kernel $\kappa$ are imposed:
\begin{assumption}\label{ass:kernel}The kernel $\kappa$ 
\begin{itemize}
\item is non-negative and decreasing, and 
\item satisifies $\int_0^T \kappa(s) ds < \infty$ for all $T > 0$.
\end{itemize}
\end{assumption}
Clearly, this assumption includes the power-law kernels $\kappa_\alpha$ for all $\alpha \in (0,1]$. The slight abuse of notation that we have introduced should not cause any confusions: $\psi_\kappa$ denotes the solution of \eqref{eq:Volterra_Riccati_general} for a general kernel $\kappa$; $\psi_\alpha$ for the power-law kernel $\kappa_\alpha$; and $\psi_1$ for the plain Heston case $\kappa_1 \equiv 1$.

The properties of the solution $\psi_\kappa$, in particular its maximal life-time, crucially depend on the nature of $R(u,w)$. As in \cite{gerhold2018moment}, we distinguish between the following cases, illustrated in Figure~\ref{fig:cases}
\begin{itemize}
\item[(A)] $R(u,0)>0$ and $\partial_w R(u,0) \ge  0$,
\item[(B)] $R(u,0)>0$ and $\partial_w R(u,0) <  0$ and $R(u,\cdot)$ has no roots,
\item[(C)] $R(u,0)>0$ and $\partial_w R(u,0) <  0$ and $R(u,\cdot)$ has positive roots,
\item[(D)] $R(u,0) \le 0$.
\end{itemize}

\begin{figure}[htb]
\includegraphics[width=0.6\textwidth]{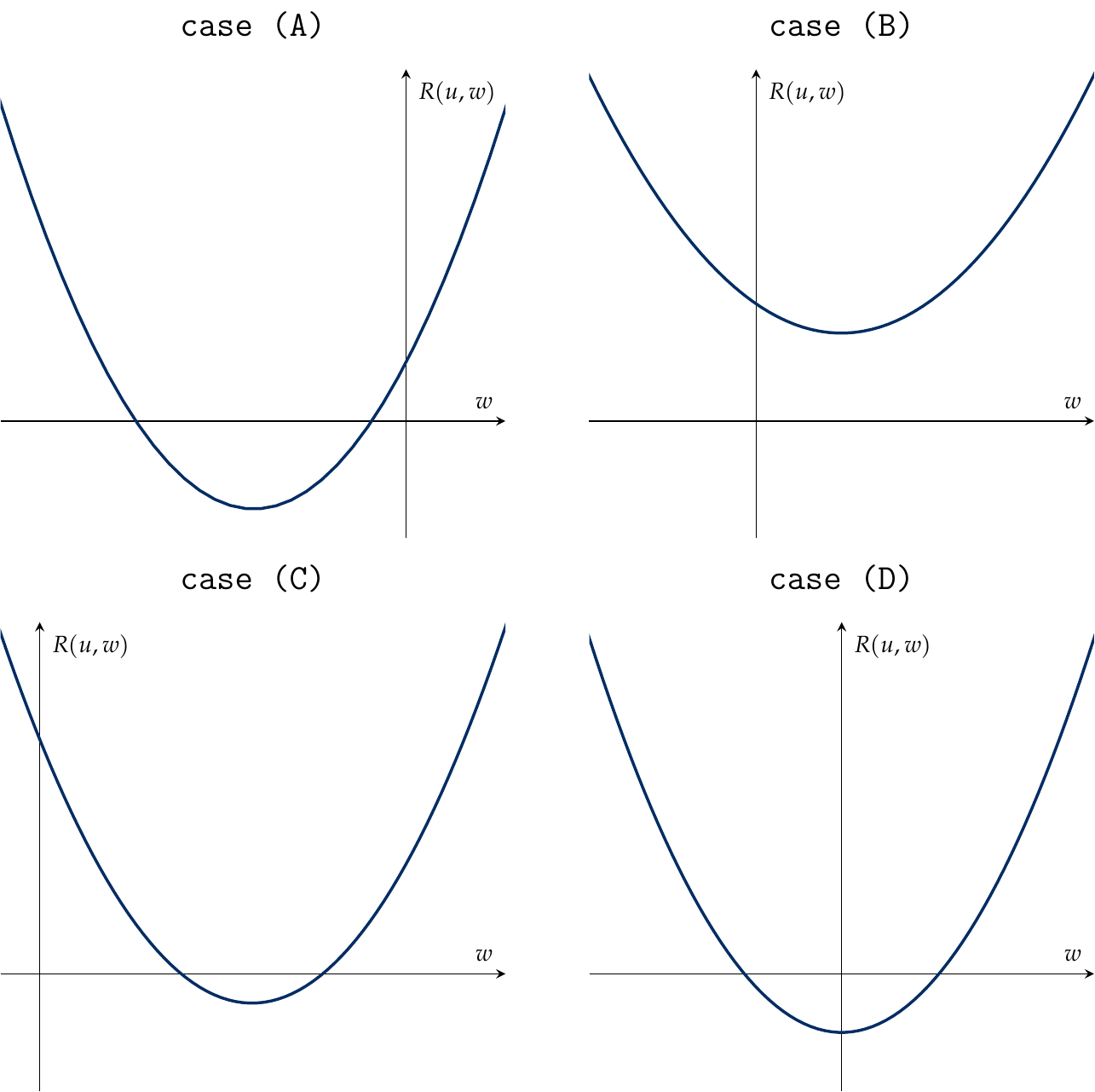}
	\caption{Schematic plot of $R(u,\cdot)$, with $u \in  \RR$ satisfying case (A), (B), (C), or (D).}\label{fig:cases}
\end{figure}

These cases can be analyzed by applying the familiar theory of quadratic equations to the polynomial $R(u,w)$. Following the notation from \cite{gerhold2018moment}, we rewrite $R(u,w)$ as
\begin{equation}\label{eq:R}
	R(u,w) = c_1(u) + c_2(u) w + \frac{\eta^2}{2} w^2,
\end{equation}
with coefficients
\begin{align*}
	c_1(u) &= \tfrac{1}{2}u(u-1), \\
	c_2(u) &= \rho \eta u - \lambda. 
\end{align*}
The discriminant of $w \mapsto R(u,w)$ is given by 
\begin{equation}\label{eq:discriminant}
	\Delta(u)  = \tfrac{1}{4} \left((\rho \eta u - \lambda)^2 - \eta^2 (u^2 - u)\right).
\end{equation}
If and only if $\Delta(u)$ is positive, $R(u,\cdot)$ has two real roots located at $\frac{1}{\eta^2}(-c_2(u) \pm 2\sqrt{\Delta(u)})$. In the case $\rho < 0$ (which is typical in applications) this leads to the following classification, which is illustrated in Figure~\ref{fig:ABCD} below:
\begin{lemma}\label{lem:caseABCD}
Suppose that $\rho < 0$ and denote the roots of $\Delta(u)$ by
\begin{equation}\label{eq:dpm}
d_\pm := \frac{\eta - 2\rho \pm \sqrt{(\eta - 2\rho)^2 + 4\lambda^2(1-\rho^2)}}{2 \eta (1 - \rho^2)}.
\end{equation}
Then $\tfrac{\lambda}{\rho \eta} < d_- < 0$; $1 < d_+$ and 
\begin{itemize}
\item $u$ satisfies case (A)  $\quad \Longleftrightarrow \quad u \le \frac{\lambda}{\rho \eta}$,
\item $u$ satisfies case (B) $\quad \Longleftrightarrow \quad u \in (\frac{\lambda}{\rho \eta},d_-) \cup (d_+,\infty)$,
\item $u$ satisfies case (C)  $\quad \Longleftrightarrow \quad u \in [d_-,0) \cup (1,d_+]$,
\item $u$ satisfies case (D)  $\quad \Longleftrightarrow \quad u \in [0,1]$.
\end{itemize}
\end{lemma}
\begin{proof}The mapping of the cases (A-D) to the corresponding intervals is based on the following observations: $R(u,0) = c_1(u)$ is negative on $[0,1]$ and strictly positive outside; $\partial_w R(u,0) = c_2(u)$ is positive for $u \le  \tfrac{\lambda}{\rho \eta}$ and strictly negative elsewhere. Finally, the discriminant $\Delta(u)$ of $w \mapsto R(u,w)$ is positive within $[d_-,d_+]$ and strictly negative outside. It remains to show the stated inequalities for $d_\pm$. Directly from \eqref{eq:dpm} it can be seen that $d_- < 0$ and that 
\[d_+ \ge \frac{2(\eta - 2\rho)}{2\eta(1 - \rho^2)} \ge \frac{1}{1 - \rho^2} > 1.\]
Finally, 
\[e_1\left(\frac{\lambda}{\rho \eta}\right) = - \frac{1}{4\rho^2} (\lambda^2 - \lambda \rho \eta) < 0\]
shows that $\tfrac{\lambda}{\rho \eta} < d_-$.
\end{proof}

If $\int_0^\infty \kappa(s) ds = \infty$, then the four cases introduced above have the following connection to the properties of $\psi_\kappa$:
\begin{itemize}
\item In cases (A) and (B), the solution $\psi_\kappa$ explodes in finite time. 
\item In cases (C) and (D), the solution $\psi_\kappa$ exists globally. 
\end{itemize}
In the power-law case $\kappa = \kappa_\alpha$ this has already been shown in \cite{gerhold2018moment}. However, our goal is not just to characterize the domains where $\psi_\kappa$ exists globally, but rather to give a more refined comparison principle between $\psi_\kappa$ and the non-rough Heston solution $\psi_1$. Such comparison results have been shown in \cite[Appendix~A]{gatheral2018affine} in the non-exploding case (C) and we will extend those arguments to cover all situations (A-D).  To formulate these results let
\begin{equation*}
	w_0(u) := -\frac{c_2(u)}{\eta^2} =  -\frac{1}{\eta^2}(\rho \eta u - \lambda) 
\end{equation*}
denote the location of the global minimum of $w \mapsto R(u,w)$, and 
\begin{equation*}
	w_{*}(u):= \frac{1}{\eta^2}\left(-c_2(u)-2\sqrt{\Delta(u)}\right)
\end{equation*}
 the location of its first root, whenever $\Delta(u) \ge 0$. The next Lemma is closely related to \cite[Lem.~A.3]{gatheral2018affine}.

%Lemma for Q -----------------------------------------------------------
\begin{lemma}\label{lem_3.1}
	Let $u \in \RR$ and $Q(u,\cdot)$ defined by
	\begin{equation}\label{eq:Q1}
		Q(u,w) := \int_0^w \frac{d\zeta}{R(u,\zeta)}, \qquad w \in  \RR.
	\end{equation}
	Furthermore, we define
	\begin{align*}
		v_1(u) &:= \frac{1}{\sqrt{-\Delta(u)}}\left(\frac{\pi}{2} - \arctan \left(\frac{c_2(u)}{2\sqrt{-\Delta(u)}} \right) \right), \\
		v_2(u) &:= \frac{1}{2\sqrt{\Delta(u)}}\log \left( \frac{c_2(u) + 2 \sqrt{\Delta(u)} }{c_2(u) -2\sqrt{\Delta(u)}} \right).
	\end{align*}
	\begin{itemize}
\item[(a)] If $u$ satisfies case (A) and $\Delta(u)<0$, the function $ Q(u,\cdot)$ maps $[0,\infty)$ onto $[0,v_1(u))$, is strictly increasing, and has an inverse $ Q^{-1}(u,\cdot) $, which maps $[0,v_1(u))$ onto $[0,\infty)$. If $\Delta(u)>0$, $v_1(u)$ has to be replaced by $v_2(u)$.
\item[(b)] If $u$ satisfies case (B), the same assertion as in (a) holds (with the restriction that only $v_1(u)$ is needed).
\item[(c)] If $u$ satisfies case (C), it holds that $w_*(u)>0$ and the function $Q(u,\cdot)$ maps $[0,w_*(u))$ onto $[0,\infty)$, is strictly increasing, and has an inverse $ Q^{-1}(u,\cdot) $, which maps $[0,\infty)$ onto $[0,w_*(u))$.
\item[(d)] If $u$ satisfies case (D), it holds that $w_*(u)<0$ and the function $Q(u,\cdot)$ maps $(w_*(u),0]$ onto $[0,\infty)$, is strictly decreasing, and has an inverse $ Q^{-1}(u,\cdot) $, which maps $[0,\infty)$ onto $(w_*(u),0]$.
\end{itemize}
\end{lemma}
\begin{remark}
While the lemma is mainly a technical result on the properties of the function $Q(u,w)$, the connection to moment explosions in the Heston model should become apparent from the fact that $T_1^*(u)$ can be written as
\[T_1^*(u) = \lim_{w \to \infty} Q(u,w) = \int_0^\infty \frac{d\zeta}{R(u,\zeta)}\]
in cases (A) and (B), cf. \cite[Sec.~6.1]{keller-ressel2011moment}.
\end{remark}

\begin{proof}
	(a) Due to the fact that the integrand $ 1/R(u,\zeta)$ is positive on $[0,\infty)$ if $u$ satisfies case (A), we can conclude that $ Q(u,\cdot)$ is strictly increasing. It just remains to show that the integral attains the limit $v_1(u)$ resp. $v_2(u)$. If $\Delta(u)<0$, we get
	\begin{align*}
		\lim_{w \rightarrow \infty} Q(u,w) &= \int_0^{\infty}\frac{d\zeta}{R(u,\zeta)} \\&= \left. \frac{1}{\sqrt{-\Delta(u)}}  \arctan \left( \frac{\eta^2 w + c_2(u)}{2\sqrt{-\Delta(u)}} \right) \right|_0^{\infty} \\&= v_1(u),
	\end{align*}
	and if $\Delta(u)>0$, we obtain
	\begin{equation*}
		\int_0^{\infty}\frac{d\zeta}{R(u,\zeta)} = \frac{1}{2 \sqrt{\Delta(u)}} \log \bigg( \frac{\eta^2 w + c_2 - 2 \sqrt{\Delta(u)}}{\eta^2 w + c_2 + 2 \sqrt{\Delta(u)}} \bigg) \Bigg|_0^{\infty} = v_2(u).
	\end{equation*}
	(b) Restricted to $\Delta(u)<0$, the proof of case (B) is analogue to (a).
	
	(c) In case (C) we can argue similar, since the integrand $ 1/R(u,\zeta)$ is positive on $[0,w_*(u))$. The assertion follows if we replace the upper limit in the above integrals by $w_*(u)$.
	
	(d) The proof of case (D) is analogue to (c), only the different sign of $R(u,\cdot)$ on $(w_*(u),0]$ has to be taken into account. 
\end{proof}
Now, we are ready to adapt the results of \cite[Appendix~A]{gatheral2018affine} to our framework.
%Theorem: comparison principle---------------------------------------------
\begin{theorem}\label{thm:comparison}
	Let $u \in \RR$.
	\begin{itemize}
	%(a)------------------------------------------------------------
	\item[(a)] If $u$ satisfies case (A), then $\psi_\kappa(\cdot,u)$ satisfies
	\begin{equation}\label{eq:bound1}
		0 \le \psi_1\left(\int_0^t \kappa(s)ds,u\right) \le \psi_\kappa(t,u), \qquad t \ge 0.
	\end{equation}
	%(b)------------------------------------------------------------
	\item[(b)] If $u$ satisfies case (B), then $\psi_\kappa(\cdot,u)$ satisfies
	\begin{equation}
		0 \le \overline \psi_1\left(\int_0^t \kappa(s)ds,u \right) \le \psi_\kappa(t,u), \qquad t \ge 0,
	\end{equation}
	where $\overline \psi_1$ is the solution of
	\begin{equation*}
		\overline \psi_1(t,u) = \int_0^t \overline R \left( u,\overline \psi_1(s,u) \right) ds, \qquad t\ge 0,
	\end{equation*}
	with $\overline R$ given by 
	\begin{equation}\label{eq:Rbar}
	\overline R(u,w) :=  \begin{cases} 
 R(u,w_0(u)) & w \le w_0(u) \\
R(u,w) & \, w > w_0(u)
\end{cases}.
\end{equation}
	%(c)------------------------------------------------------------
	\item[(c)] If $u$ satisfies case (C), then $\psi_\kappa(\cdot,u)$ exists globally and satisfies
	\begin{equation}\label{bounds_case_c}
		0 \le \psi_\kappa(t,u) \le \psi_1\left(\int_0^t \kappa(s)ds,u\right) \le w_{*}(u), \qquad t \ge 0.
	\end{equation}
	%(d)------------------------------------------------------------
	\item[(d)] If $u$ satisfies case (D), then $\psi_\kappa(\cdot,u)$ exists globally and satisfies
	\begin{equation}
		w_{*}(u) < \psi_1\left(\int_0^t \kappa(s)ds,u \right) \le \psi_\kappa(t,u) < 0, \qquad t \ge 0.
	\end{equation}
	\end{itemize}
\end{theorem}
\begin{remark}The function $w \mapsto \overline R(u,w)$ introduced in \eqref{eq:Rbar} should be interpreted as \emph{increasing lower envelope} of $w \mapsto R(u,w)$, i.e., the largest increasing function bounding it from below.
\end{remark}
%proof------------------------------------------------------------
\begin{proof} By \cite[Thm.~12.11]{gripenberg1990volterra} equation \eqref{eq:Volterra_Riccati_general} has a continuous local solution $\psi_\kappa(\cdot,u)$ on some non-empty time interval $[0,T_\kappa(u))$. In addition, $\psi_\kappa(\cdot,u)$ can be continued up to (but not beyond) a maximal interval of existence $[0,\hat{T}_\kappa(u))$, which is open to the right, and for $\hat{T}_\kappa(u) < \infty $ it holds that 
\begin{equation}\label{eq:limsup_psi_alpha}
 \limsup_{t \to \hat{T}_\kappa(u) } \psi_\kappa(t,u) = \infty.
\end{equation}
In particular, this means that $\hat{T}_\kappa(u)$ can be written as 
\[\hat{T}_\kappa(u) = \sup\{t > 0: \psi_\kappa(t,u) < \infty\},\]
consistent with \eqref{eq:Tstar}.

%(a)------------------------------------------------------------
(a) Let $u$ satisfy case (A). Recall the Riccati equation in the non-rough Heston model:
	\begin{equation}\label{eq:re2}
		\partial_t \psi_1(t,u) = R(u,\psi_1(t,u)).
	\end{equation}
	We claim that its solution satisfies
	\begin{equation}\label{eq:psiQ-1}
		Q(u,\psi_1(t,u)) = t, \qquad \forall t \in [0,\hat{T}_1(u)),
	\end{equation} 
	where $Q$ is given by \eqref{eq:Q1}. Dividing by $R(u,\psi_1(t,u))$ and integrating both sides of \eqref{eq:re2} yields
	\begin{equation*}
		\int_0^t \frac{\partial_s \psi_1(s,u)}{R(u,\psi_1(s,u))} ds = t.
	\end{equation*}
	Now we substitute $ \eta = \psi_1(s,u) $, $d\eta = \partial_s \psi_1(s,u) ds $, and get
	\begin{equation}\label{eq:subst}
		\int_0^{\psi_1(t,u)} \frac{d\eta}{R(u,\eta)} = t,
	\end{equation}
	which verifies \eqref{eq:psiQ-1}.
	
	 We remember that for $u $ satisfying case (A), the function $R(u,\cdot)$ is positive and increasing on $[0,\infty)$. Since the kernel $\kappa$ is decreasing, we can deduce the following inequality
	\begin{align}\label{eq: v_ineq}
		\psi_\kappa(t,u) &= \int_0^t \kappa(t-s) R(u,\psi_\kappa(s,u)) ds \nonumber \\ &\ge \int_0^t \kappa(T-s) R(u,\psi_\kappa(s,u)) ds =: v(t,T), 
	\end{align}
	for $0 \le t \le T < \hat{T}_\kappa(u)$. It is easily seen that the above defined function $v(t,T)$ has the boundary values
	\begin{align}\notag
	    v(0,T) &= 0, \\ \label{eq: v=psi_alpha}
		v(t,t) &= \psi_\kappa(t,u), 
	\end{align}
	and, since $w \mapsto R(u,w)$ is increasing on $[0,\infty)$, satisfies the differential inequality
	\begin{equation}\label{eq:v_dif_ineq}
		\partial_t v(t,T) = \kappa(T-t) R(u,\psi_\kappa(t,u)) \ge \kappa(T-t) R(u,v(t,T)).
	\end{equation}
	Now, we can use a standard comparison principle for differential equations (see e.g. Chapter II, \S \ 9 in \cite{browder1998ordinary}) to obtain
	\begin{equation}\label{eq:comp_princ}
		v(t,T)\ge r(t,T),
	\end{equation}
	with $r(t,T)$, being the solution of
		\begin{equation}\label{eq:r1}
		\partial_t r(t,T) = \kappa(T-t) R(u,r(t,T)).
	\end{equation}
	Note that this differential equation differs from \eqref{eq:re2} only by the factor $\kappa(T-t)$. Thus, if we divide by $R(u,r(t,T))$, integrate both sides up to $T$ and substitute analogue to the Heston case in \eqref{eq:subst} with $ \eta = r(t,T) $, $d\eta = \partial_t r(t,T) dt $, we get
	\begin{equation}\label{eq:psiQ-2}
		Q(u, r(T,T)) = \int_0^{r(T,T)} \frac{d\eta}{R(u,\eta)} = \int_0^T \kappa(T-t) dt = \int_0^T \kappa(t) dt ,	\end{equation}
		with $T < T_\kappa^*(u)$. Applying $Q^{-1}(u,\cdot)$ to \eqref{eq:psiQ-1} and \eqref{eq:psiQ-2}, it holds that
	\begin{align}\notag
	\psi_1(t,u) &= Q^{-1}(t,u), \\ \label{eq:r=Q}
		r(t,t) &= Q^{-1}\left( u,\int_0^t \kappa(s) ds \right),
	\end{align}
	and we can deduce
	\begin{equation}\label{eq:r=psi1}
		r(t,t) = \psi_1 \left( u,\int_0^t \kappa(s) ds \right), \qquad t \in \left[0,\hat{T}_\kappa(u)\right).
	\end{equation}
	The inequalities \eqref{eq: v_ineq} and \eqref{eq:comp_princ} finally yield
	\begin{equation}\label{eq:kappa_compare}
		\psi_\kappa(t,u) = \lim_{T \downarrow t} v(t,T) \ge \lim_{T \downarrow t} r(t,T) = \psi_1\left( u, \int_0^t \kappa (s)ds \right), \qquad t \in \left[0,\hat{T}_\kappa(u)\right).
	\end{equation}
	
	% PART (b) of the proof---------------------------------------------------	
(b) If $u$ satisfies case (B), the inequality \eqref{eq: v_ineq} still holds. However, we cannot argue as in \eqref{eq:v_dif_ineq}, because the function $R(u,\cdot)$ is decreasing on $[0,w_0(u))$. To circumvent this obstacle, we use the adjusted function $\overline R(u,w)$ from \eqref{eq:Rbar} and conclude the inequalities
\begin{align*}
	\partial_t v(t,T) &= \kappa(T-t) R(u,\psi_\kappa(t,u)) \\
	&\ge \kappa(T-t) \overline R(u,\psi_\kappa(t,u)) \\
	&\ge \kappa(T-t) \overline R(u,v(t,T)),
\end{align*}
for all $0\le t \le T < \hat{T}_\kappa(u)$. From this point we can proceed as in (a) with the function $ \overline r(t,T)$, being the solution of
\begin{equation*}
		\partial_t \overline r(t,T) = \kappa(T-t) \overline R(u, \overline r(t,T)).
\end{equation*}

% PART (c) of the proof

(c) Let $u$ be satisfying case (C) and set
\begin{equation}\label{eq:T_tilde}
	\widetilde{T}_\kappa(u) := \inf \left\{ t \in \left(0,\hat{T}_\kappa(u)\right) : \psi_\kappa(t,u) = w_{*}(u) \text{ or } \psi_\kappa(t,u) = 0\right\}.
\end{equation}
Due to the behavior of the function $R(u,\cdot)$ in case (C), and because of \eqref{eq:Volterra_Riccati_general} we can conclude
\begin{equation}\label{eq:psi_alpha >0}
	\psi_\kappa(t,u) > 0, \qquad \forall t \in \left(0,\widetilde{T}_\kappa(u)\right).
\end{equation}
This clearly indicates that $\psi_\kappa(\cdot,u)$ is increasing for $t \in \left(0,\widetilde{T}_\kappa(u)\right)$, and therefore the upper bound in \eqref{eq:T_tilde} is always hit before the lower bound.

Now we can continue similarly to (a): \\
Considering $0 \le t \le T \le \widetilde{T}_\kappa(u)$, it can be seen that the inequality \eqref{eq: v_ineq} is satisfied. In \eqref{eq:v_dif_ineq}, however, the inequality sign has to be reversed, since  $R(u,\cdot)$ is decreasing on $[0,w_*(u))$. Therefore, the solution $r(t,T)$ of \eqref{eq:r1} satisfies 
\begin{equation}\label{eq:r>=v}
	r(t,T) \ge v(t,T), \qquad  0 \le t \le T \le \widetilde{T}_\kappa(u).
\end{equation}
Using \eqref{eq: v_ineq},\eqref{eq: v=psi_alpha}, \eqref{eq:r=psi1}, and \eqref{eq:r>=v}, we obtain
\begin{equation}\label{eq:psi1 ge psi_alpha}
	\psi_1 \left( u,\int_0^t \kappa(s)ds \right) = \lim_{T \downarrow t} r(t,T) \ge \lim_{T \downarrow t} v(t,T) = \psi_\kappa(t,u),
\end{equation}
for all $t \in [0,\widetilde{T}_\kappa(u))$. By means of \eqref{eq:r=Q}, \eqref{eq:r=psi1}, and Lemma \ref{lem_3.1} (c), this implies that
\begin{equation}\label{eq:psi_alpha le psi1}
	\lim_{t \to \widetilde{T}_\kappa(u)} \psi_\kappa(t,u) \le \psi_1  \left( u,\int_0^{\widetilde{T}_\kappa(u)} \kappa(s)ds \right) < w_*(u).
\end{equation}
Considering \eqref{eq:T_tilde}, we now obtain $ \widetilde{T}_\kappa(u) = \hat{T}_\kappa(u)$, i.e. the bounds \eqref{bounds_case_c} hold for all $t \in [0,\hat{T}_\kappa(u))$ and we have
\[\lim_{t \to \hat{T}_\kappa(u))} \psi_\kappa(t,u)  \in [0,w_*(u)].\]
If $\hat{T}_\kappa(u)) < \infty$, this is a contradiction to \eqref{eq:limsup_psi_alpha}, and we conclude that $\hat{T}_\kappa(u) = \infty$.

(d) The proof of the bounds in case (D) is analogous to (c) with the following adaptations: The inequality sign in \eqref{eq: v_ineq} has to be reversed, since the function $R(u,\cdot)$ is negative on $(w_*(u),0]$. Thus, in contrast to (c), the inequality sign of \eqref{eq:v_dif_ineq} remains. It follows that the inequality signs of \eqref{eq:psi_alpha >0}, \eqref{eq:r>=v}, \eqref{eq:psi1 ge psi_alpha}, and \eqref{eq:psi_alpha le psi1} have to be reversed, and the proof is complete.
\end{proof}
% end proof------------------------------------------------------------
\subsection{First consequences}
We state two immediate corollaries from Theorem~\ref{thm:comparison}. The first generalizes \cite[Thm.~2.4]{gerhold2018moment} from power-law kernels  to a large class of other kernels.
\begin{corollary}Let $\kappa$ be a kernel satisfying Assumption~\ref{ass:kernel} and with $\int_0^\infty \kappa(s)ds = \infty$. Then $\hat{T}_\kappa(u)$ is finite if and only if $u$ satisfies case (A) or (B), and it is infinite if and only if $u$ satisfied case (C) or (D). In particular, the set $\set{u \in \RR: \hat{T}_\kappa(u) < \infty}$ is independent of $\kappa$.
\end{corollary}
\begin{proof}In case (A) it is known that $\psi_1(t,u)$ blows up in finite time, cf. \cite{andersen2007moment, keller-ressel2011moment}. Since $\int_0^\infty \kappa(s) ds = \infty$ and, by Theorem~\ref{thm:comparison}, 
\[\psi_1\left(\int_0^t \kappa(s) ds, u \right) \le \psi_\kappa(t,u)\]
for all $t \ge 0$, also $\psi_\kappa(t,u)$ must blow up in finite time.\\
In case (B), $\overline \psi_1(t,u)$ has to be used instead of $\psi_1(t,u)$. It can be seen by direct calculation that also $\overline \psi_1(t,u)$ blows up in finite time, see also Lemma~\ref{lem:Tbar} below. In cases (C) and (D) Theorem~\ref{thm:comparison} shows global existence of $\psi_\kappa(t,u)$, i.e., no finite-time blow-up can take place.
\end{proof}
In many cases of interest, the time-change $T \mapsto \int_0^T \kappa(s)ds$ contracts time for small $T$ up to a time $\mathfrak{T}_\kappa$; see  Figure~\ref{fig:distortion}. This allows to reformulate Theorem~\ref{thm:comparison} without time-change, at the expense weakening the inequalities.
 %COR--------------------------------------------------------------------------
 \begin{corollary}\label{cor:direct_comp}
 	Suppose that $\kappa$ is strictly decreasing and there exists $t_* \in (0,\infty)$ with $\kappa(t_*) = 1$. Then, there is a unique solution $\mathfrak{T}_\kappa \in (0, \infty)$ of 
	\begin{equation}\label{eq:fixedpoint}
	T = \int_0^T \kappa(s) ds
	\end{equation}
	and the following holds: 
	\begin{itemize}
	%------------------------------------------------------------
	\item[(a)] If $u$ satisfies case (A), it holds that
	\begin{equation*}
		\psi_1(t,u) \le \psi_{\kappa}(t,u), \qquad \forall t \le \mathfrak{T}_\kappa.
	\end{equation*}
	%------------------------------------------------------------
	\item[(b)]If $u$ satisfies case (B), it holds that
	\begin{equation*}
		\overline \psi_1(t,u) \le \psi_{\kappa}(t,u), \qquad \forall t \le \mathfrak{T}_\kappa.
	\end{equation*}
	%------------------------------------------------------------
\end{itemize}
\end{corollary}

 %--- check!
\begin{proof}
Under the given assumptions, the function $t \mapsto \int_0^t \kappa(s) ds$ starts at $t= 0$, is increasing, strictly concave, and has derivative one at $t_* \in (0,\infty)$. It is obvious that this implies the existence of a unique fixed point $\mathfrak{T}_\kappa \in (0,\infty)$, i.e., of a unique solution of \eqref{eq:fixedpoint}. 
Moreover, 	\begin{equation*}
		t \le \int_0^t \kappa(s) ds
			\end{equation*}
			must hold for all $t \le \mathfrak{T}_\kappa$. Since $ \psi_1(\cdot,u) $ is strictly increasing in cases (A-C) (see Chapter 2 in \cite{gatheral2006volatility}), we obtain from Theorem \ref{thm:comparison} that
			\[\psi_1(t,u) \le \psi_1\left(\int_0^t \kappa(s)ds,u \right) \le \psi_\alpha(t,u)\]
			for $t \le \mathfrak{T}_\kappa$, completing case (a). The proof of (b) is analogue.
\end{proof}

	\begin{figure}[htb]
\includegraphics[width=0.6\textwidth]{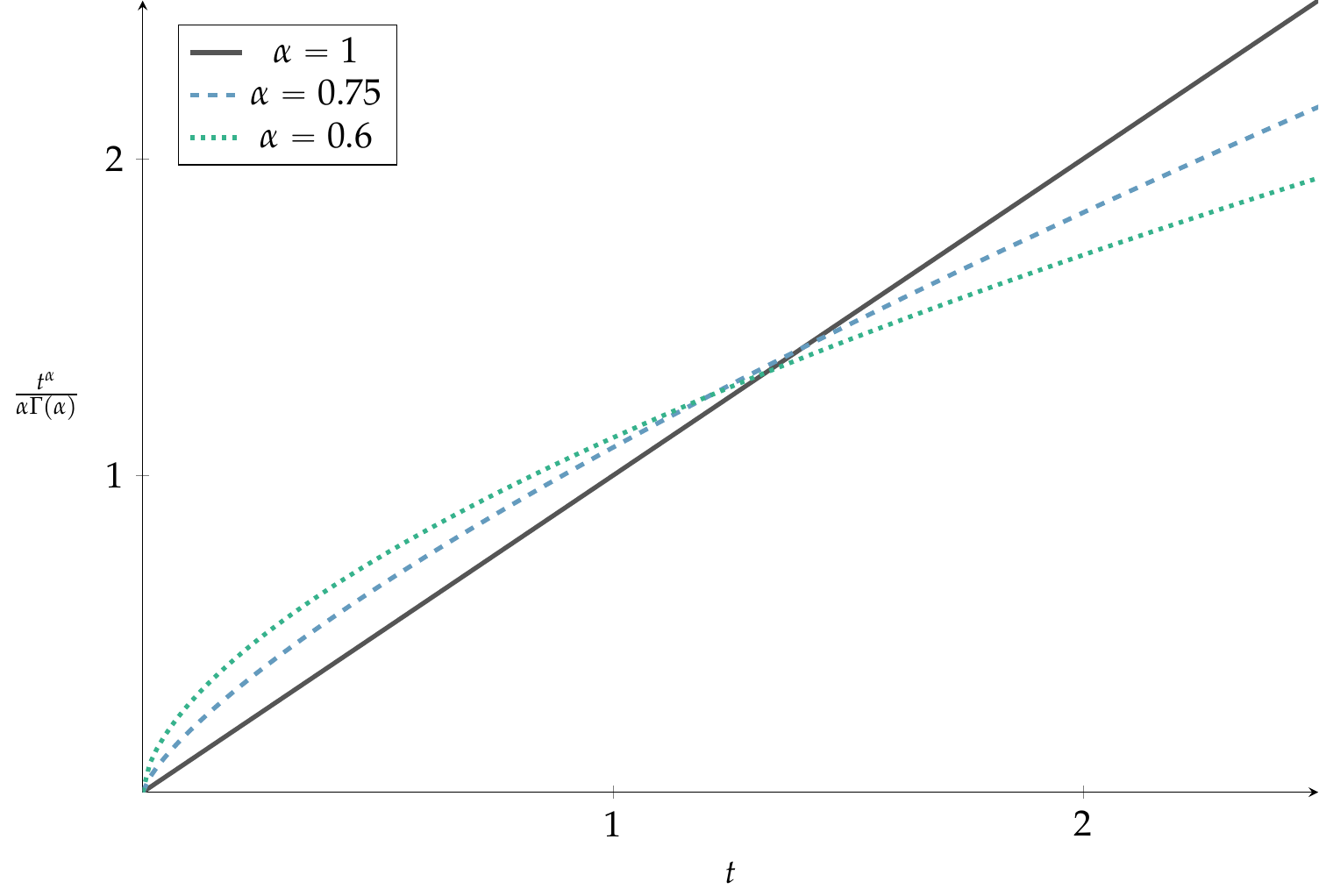}
	\caption{The graph of the time-change $ \int_0^t \kappa_{\alpha} (s) ds = \frac{t^{\alpha}}{\alpha \Gamma(\alpha)} $ for different $ \alpha $.}\label{fig:distortion}
	\end{figure}

In the rough Heston model with power-law kernel $\kappa_\alpha(t) = \frac{1}{\Gamma(\alpha)}t^{\alpha-1}$ the relevant time-change can be easily computed and is given by $\int_0^t \kappa_\alpha(s)ds = \frac{t^\alpha}{\alpha \Gamma(\alpha)}$. The kernel $\kappa_\alpha$ also satisfies the requirements of Corollary~\ref{cor:direct_comp} and the solution of \eqref{eq:fixedpoint} is given by $\mathfrak{T}_\alpha = (\alpha \Gamma(\alpha))^{1/(\alpha-1)}$. An illustration is given in Figure~\ref{fig:distortion}
    
\section{Comparison of moment explosion times}
%\subsection{Moment Explosion Times}
    In this section we study the temporal evolution of moments $\EE[S_t^u]$ of the price process in the rough Heston model \eqref{eq:rough_Heston}.
    Whereas in the Black-Scholes model moments of all orders exist for all maturities, it is well-known that moments in stochastic volatility models can become infinite at a certain time (see e.g. \cite{friz2010encyclopedia}). Recall from \eqref{eq:moment_explosion} the definition of the \emph{time of moment explosion} $T_{\alpha}^*(u) = \sup\{t \ge 0: \EE[S_t^u] < \infty\}$ for the moment of order $u$ in the rough Heston model with index $\alpha \in (\frac{1}{2},1]$. In the Heston case, $T_1^*(u)$ is known explicitly and given by 
\begin{equation}\label{eq:T_explicit}
	T_1^*(u) 
%	&= \begin{cases} 
% \int_0^{\infty} \frac{dw}{R(u,w)}, & R(u,\cdot)\text{ has no roots on }[0,\infty),  \\
%\infty, &  \text{else},
%\end{cases} \notag \\ \label{eq: expl_time_Heston}
= \begin{cases} 
\frac{1}{\sqrt{-\Delta(u)}}\left(\frac{\pi}{2} - \arctan \left(\frac{c_2(u)}{2\sqrt{-\Delta(u)}} \right) \right), & \Delta(u)<0 , \text{ (A) or (B)}\\
\frac{1}{2\sqrt{\Delta(u)}}\log \left( \frac{c_2(u) + 2 \sqrt{\Delta(u)} }{c_2(u) -2 \sqrt{\Delta(u)}} \right), &  \Delta(u) > 0, c_2(u) > 0, \text{ (A)}\\
\infty, &  \Delta(u) \ge 0, c_2(u)<0, \text{ (C) or (D)},
\end{cases}
\end{equation}
see \cite{andersen2007moment, keller-ressel2011moment}. For $\alpha < 1$ in contrast, $\psi_\alpha(t,u)$ is not known explicitly and therefore also no explicit expression for $T_\alpha^*(u)$ can be derived. As discussed in the introduction, an upper bound, a lower bound an an approximation method (valid in case (A)) for $T_\alpha^*(u)$ have been derived in \cite{gerhold2018moment}. Here, we obtain an alternative upper bound of $T_{\alpha}^*(u)$ in terms of $T_1^*(u)$ as a direct consequence of Theorem~\ref{thm:comparison}:
%Heston upper boundary Theorem--------------------------------------------
\begin{theorem}\label{thm:heston_upper_bound}
	Let $u \in \RR$, such that case (A) holds. Then the blow-up time $T_{\alpha}^*(u)$ satisfies
	\begin{equation}\label{eq:upper_bound}
		T_{\alpha}^*(u) \le \left(\alpha \Gamma(\alpha) T_1^*(u) \right)^{1/\alpha}.
	\end{equation}
	 If, in addition, $T_1^*(u) \le \mathfrak{T}_\alpha$, where $\mathfrak{T}_\alpha = (\alpha \Gamma(\alpha))^{1/(\alpha -1)}$ (as in Corollary \ref{cor:direct_comp}), then
	  \begin{equation}\label{eq:T_H le T_rH}
	  	T_{\alpha}^*(u) \le T_1^*(u).
	  \end{equation}
	  The two inequalities also hold in case (B) when $T_1^*(u)$ is replaced by $ \overline T_1^*(u)$. In cases (C) and (D) it holds that $T_\alpha^*(u) = \infty$.
\end{theorem}
\begin{proof} By Theorem~\ref{thm:comparison}, we know that 
\[\psi_1\left(\int_0^t \kappa_\alpha(s) ds, u \right) \le \psi_\alpha(t,u)\]
for all $t \in [0,\hat T_\alpha(u))$. Clearly, the right hand side must blow-up before the left hand side, and therefore the blow-up time of $\psi_1(\int_0^t \kappa_{\alpha}(s) ds,u)$ represents an upper bound $T_\alpha^+(u)$ of $\hat T_\alpha(u)$. Since the blow-up time $T_1^*(u)$ of $\psi_1(t,u)$ is known,  we can determine $T_\alpha^+(u)$ by solving the equation
	\begin{equation*}
		\int_0^{T_\alpha^+(u)} \kappa_{\alpha}(s) ds = T_1^*(u),
	\end{equation*}
	which leads us to
	\begin{equation*}
		T_\alpha^+(u) = (\alpha \Gamma(\alpha)T_1^*(u))^{1/\alpha}.
	\end{equation*}
By Theorem~\ref{thm:mgf2} $T_{\alpha}^*(u) = \hat{T}_{\alpha}(u)$, which proves \eqref{eq:upper_bound}. Using the same argument as in the proof of Corollary~\ref{cor:direct_comp}, we obtain that 
	\begin{equation*}
		T_\alpha^+(u) \le \int_0^{T_\alpha^+(u)} \kappa_{\alpha}(s) ds = T_1^*(u),
	\end{equation*}
	as long as $T_1^*(u) \le \mathfrak{T}_\alpha$, and \eqref{eq:T_H le T_rH} follows. The proof of case (B) is analogue.
\end{proof}
	
	The explicit form of $T_1^*(u)$ has been given in \eqref{eq:T_explicit}. The bound $\overline T_1^*(u)$, relevant in case B, can also be computed explicitly:
	\begin{lemma}\label{lem:Tbar}
	For $u$ in in case (B), the explosion time $\overline T_1^*(u)$ of $\overline \psi_1(t,u)$ is given by
	\begin{equation}\label{eq:Tbar}
	\overline T_1^*(u) = \frac{1}{\sqrt{-\Delta(u)}}\left(\frac{\pi}{2} - \frac{c_2(u)}{2\sqrt{-\Delta(u)}} \right).
	\end{equation}
\end{lemma}
\begin{remark}
Direct comparison of \eqref{eq:T_explicit} and \eqref{eq:Tbar} shows that the difference between $T_1^*(u)$ and $\overline T_1^*(u)$ can be reduced to the linearization $\arctan(x) \sim x$ of the arctangent around zero. This observation can be used to show that for $\rho < 0$, the piecewise defined function
	\[\widetilde{T}^*_1(u) := \begin{cases}T_1^*(u), &\qquad u \le \frac{\lambda}{\rho\eta} \\ \overline T_1^*(u), &\qquad u \in (\lambda/(\rho \eta),d_-) \cup (d_+,\infty) \end{cases}\]
	is twice continuously differentiable at the cut-point $u = \lambda/(\rho\eta)$, i.e. $T_1^*(u)$ transitions smoothly into $\overline T_1^*(u)$ at the boundary between case (A) and (B). See Figure~\ref{fig:ABCD} for an illustration.
\end{remark}

\begin{figure}[htb]
\includegraphics[width=0.75\textwidth]{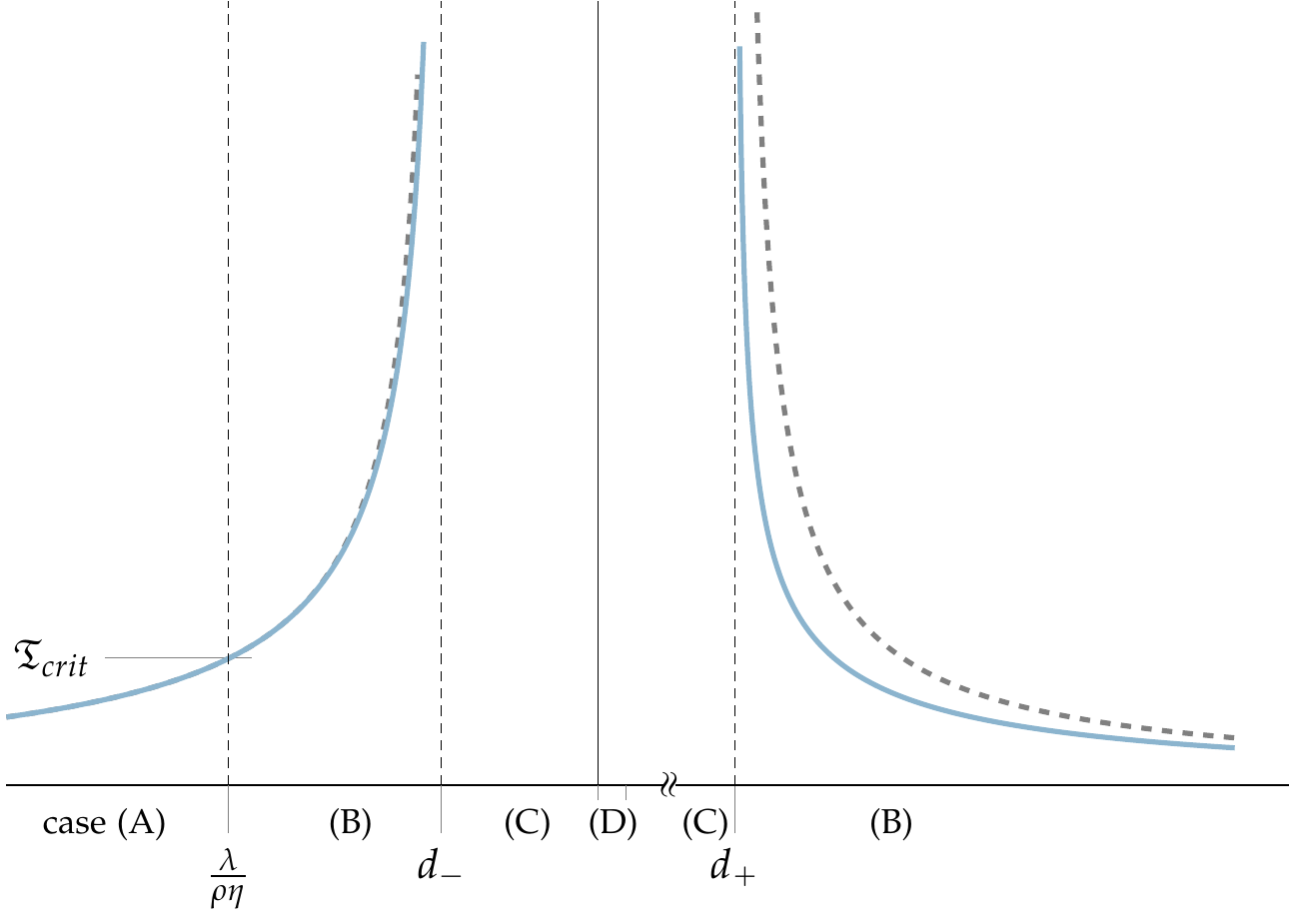}
	\caption{Illustration of the boundaries between cases (A-D) in the negative-leverage case $\rho <0$ (see Lemma~\ref{lem:caseABCD}), of the classic Heston explosion time $T_1^*(u)$ (blue solid), and of the auxiliary explosion time $\overline T_1^*(u)$ (grey dashed). The time $\mathfrak{T}_\text{crit}$ introduced in \eqref{eq:Tfrak} is also indicated.}\label{fig:ABCD}
\end{figure}

\begin{proof}From the differential equation $\frac{\partial}{\partial_t}\overline \psi(t,u) = \overline R(u,\overline \psi(t,u))$ with initial condition $\overline \psi(0,u) = 0$, we derive that
\[t = \int_0^{\overline \psi(t,u)}  \frac{d\eta}{\overline R(u,\eta)}\]
Sending $t \to \overline T_1^*(u)$ and taking into account the definition of $\overline R(u,w)$ in \eqref{eq:Rbar} yields
\[\overline T_1^*(u) = \int_0^\infty \frac{d\eta}{\overline R(u,\eta)} = \int_0^{w_0(u)}  \frac{d\eta }{R(u,w_0(u))}+ \int_{w_0(u)}^\infty \frac{d\eta }{R(u,\eta)}.\]
A primitive of $ \eta \mapsto 1/R(u,\eta)$ is given by 
\[F(w) := \frac{1}{\sqrt{-\Delta(u)}} \arctan\left(\frac{\eta^2 w + c_2(u)}{2 \sqrt{-\Delta(u)}} \right).\]
Note that $F(w_0(u)) = 0$ and $w_0(u)/R(u,w_0(u)) = c_2(u)/(2 \Delta(u))$. Hence,
\[\overline T_1^*(u) = \frac{c_2(u)}{2 \Delta(u)} + F(\infty) - F(w_0(u)) = \frac{1}{\sqrt{-\Delta(u)}}\left(\frac{\pi}{2} - \frac{c_2(u)}{2 \sqrt{-\Delta(u)}} \right)\]
as claimed. 
\end{proof}

\section{Comparison of moments}
For the comparison of moments, we fix the parameters $\rho, \lambda$ and $\eta$ of both the rough and the non-rough Heston model, but not $\theta(.)$. Instead we assume that $\theta(.)$ is determined by calibrating each model to a fixed forward variance curve; see Section~\ref{sec:calibration}. We write 
\[\Phi_\alpha(t,u) = \E{S_t^u}, \qquad \text{$S$ is $\alpha$-rough Heston}\]
for the moment generating function in dependency on $\alpha$ and $(t,u)$. In addition, we set
\begin{equation}\label{eq:KL}
\begin{split}
K_\alpha(t) &:= \int_0^t \kappa_\alpha(s) ds = t^\alpha / (\alpha \Gamma(\alpha))\\
L_{\alpha,\lambda}(t) &= \frac{1}{\lambda} \int_0^t r_{\alpha,\lambda}(s) ds =  \int_0^t s^{\alpha- 1} E_{\alpha,\alpha}(-\lambda s^\alpha) ds,
\end{split}
\end{equation}
where $r_{\alpha,\lambda}$ is the $\lambda$-resolvent kernel from \eqref{eq:resolvent}. Note that both are continuous, positive, strictly increasing functions (`time-changes') with infinite derivative at $t= 0$. However, $K_\alpha(t) \to \infty$ as $t \to \infty$, while $L_{\alpha,\lambda}(t) \to \tfrac{1}{\lambda}$. For both time-changes, there exists a unique solution in $(0,\infty)$ of
\[K_\alpha(t) = t, \quad \text{and} \quad L_{\alpha,\lambda}(t) = t,\]
which we denote by $\mathfrak{T}_\alpha$ and $\mathfrak{T}_{\alpha,\lambda}$ respectively; see also Cor.~\ref{cor:direct_comp} where these times are introduced for a generic kernel $\kappa$. It is easy to calculate that $\mathfrak{T}_\alpha = (\alpha \Gamma(\alpha))^{1/(\alpha-1)}$, while $\mathfrak{T}_{\alpha,\lambda}$ cannot be given in explicit form.

\begin{theorem}\label{thm:mgf_comparison}
Let $\alpha \in (\tfrac{1}{2},1)$, $\rho < 0$ and let $\Phi_1(t,u)$ and $\Phi_\alpha(t,u)$ be the moment generating functions of a non-rough and a rough Heston model, which are calibrated to the same forward variance curve $\xi$. Then 
\[\Phi_1(t,u) \le \Phi_\alpha(t,u)\]
holds 
\begin{itemize}
\item[(a)] for all $u \le \lambda/(\rho \eta)$ and $t \le \mathfrak{T}_\alpha$, and
\item[(b)] for all $u \in (\lambda/(\rho \eta), 0]$ and $t \le \mathfrak{T}_{\alpha,\lambda}$.
\end{itemize}
\end{theorem}

This theorem allows the direct comparison of the moment generating functions of rough and non-rough Heston models for small enough times $t$ and negative $u$. To extend the result to all $t$, we have to make a monotonicity assumption on the forward variance curve and use the time-changes introduced in \eqref{eq:KL}.

\begin{corollary}\label{cor:mgf_comparison}
Let the assumptions of Theorem~\ref{thm:mgf_comparison} hold. In addition, assume that the forward variance curve is flat or increasing. Then
\begin{itemize}
\item[(a)] for all $u \le \lambda/(\rho \eta)$ it holds that 
\[\Phi_1\Big(t \wedge K_\alpha(t),u\Big) \le \Phi_\alpha(t,u)\]
\item[(b)] for all $u \in (\lambda/(\rho \eta),0]$ it holds that 
\[\Phi_1\Big(t \wedge L_{\alpha,\lambda}(t),u\Big) \le \Phi_\alpha(t,u),\]
with $K_\alpha$ and $L_{\alpha,\lambda}$ as in \eqref{eq:KL}.
\end{itemize}
\end{corollary}

\begin{proof}[Proof of Theorem~\ref{thm:mgf_comparison}]
Our starting point is the representation \eqref{eq:mgf_alternative} of the moment generating function in a (rough or non-rough) Heston model calibrated to a forward variance curve $\xi$. From this representation, it is clear that the statement 
\begin{equation}\label{eq:mgf_comp_generic}
\Phi_1(t',u) \le \Phi_\alpha(t,u)
\end{equation}
for some $t,t' \in \Rplus$ is equivalent to 
\begin{equation}\label{eq:mgf_comp_explicit}
\int_0^{t'} \xi(t'-s) R_0(u,\psi_1(s,u)) ds  \le \int_0^{t} \xi(t-s) R_0(u,\psi_\alpha(s,u)) ds,
\end{equation}
where we set
\begin{equation}\label{eq:R0}
R_0(u,w) = R(u,w) + \lambda w = \frac{1}{2}(u^2- u) + \rho \eta u w + \frac{\eta^2}{2}w^2.
\end{equation}
From Corollary~\ref{cor:direct_comp}a we obtain that $\psi_1(s,u) \le \psi_\alpha(s,u)$ for all $s \le \mathfrak{T}_\alpha$ and $u \le \lambda / (\rho \eta)$. Since $w \mapsto R_0(u,w)$ is increasing for positive arguments, \eqref{eq:mgf_comp_explicit} follows with $t'=t$ and part (a) of the Theorem is shown.\\
For $u \in (\lambda/(\rho \eta),0]$, we are in the domain of case (B) or (C). Instead of using Corollary~\ref{cor:direct_comp}b (which does not allow direct comparison with the non-rough Heston model) we transform the Volterra-Riccati integral equation \eqref{eq:Volterra_Riccati} using the resolvent kernel $r_{\alpha,\lambda}$ from \eqref{eq:resolvent}. Using the convolution notation $f \star g = \int_0^t f(t-s)g(s)ds$, the resolvent kernel is characterized by the property
\[\lambda \kappa_\alpha - r_{\alpha,\lambda} = \lambda r_{\alpha,\lambda} \star \kappa_\alpha,\]
see e.g. \cite[Ch.~2]{gripenberg1990volterra}. Convolving $\psi_\alpha$ (and suppressing its dependency on $u$) with $r_{\alpha,\lambda}$, we obtain
\[r_{\alpha,\lambda} \star \psi_\alpha =  r_{\alpha,\lambda} \star \kappa_\alpha \star R(u,\psi_\alpha) = \kappa_\alpha \star R(u,\psi_\alpha) - \tfrac{1}{\lambda} r_{\alpha,\lambda} \star R(u,\psi_\alpha).\]
Subtracting this from the Volterra integral equation $\psi_\alpha = \kappa_\alpha \star R(u,\psi_\alpha)$ we obtain
\begin{equation}\label{eq:volterra_resolvent}
\psi_\alpha(t,u) = \tfrac{1}{\lambda} \int_0^t r_{\alpha,\lambda}(t-s)R_0(u,\psi_\alpha(s,u))ds,
\end{equation}
another Volterra integral equation for $\psi_\alpha$, now involving the kernel $\tfrac{1}{\lambda} r_{\lambda,\alpha}$. This kernel satisfies Assumption~\ref{ass:kernel} and an application of Corollary~\ref{cor:direct_comp} yields that $\psi_1(s,u) \le \psi_\alpha(s,u)$ for all $s \le \mathfrak{T}_{\alpha,\lambda}$. Note that the domain of case (A) has to be determined relative to $R_0(u,w)$, which now includes all $u \le 0$. The remaining proof of part (b) follows by repeating the arguments of part (a).
\end{proof}

\begin{proof}[Proof of Corollary~\ref{cor:mgf_comparison}]
Assume $u \le \lambda / (\rho \eta)$ and observe that
\[t \wedge K_\alpha(t) = \begin{cases} t, &\qquad t < \mathfrak{T}_\alpha\\ K_\alpha(t) &\qquad t \ge \mathfrak{T}_\alpha.\end{cases}\]
Thus, for $t < \mathfrak{T}_\alpha$ the claim of the Corollary is already covered by Theorem~\ref{thm:mgf_comparison} and it remains to treat the case $t \ge \mathfrak{T}_\alpha$. From the concavity of $K_\alpha$, it follows that
\begin{equation}\label{eq:K_concave}
K_\alpha(t) - K_\alpha(r) \le \kappa_\alpha(r) (t-r) \le t-r
\end{equation}
for all $\mathfrak{T}_\alpha \le r \le t$; note that $\kappa_\alpha(r) \le 1$ for any such $r$. Moreover, from Theorem~\ref{thm:comparison} we know that $\psi_1(K_\alpha(t),u) \le \psi_\alpha(t,u)$ for all $t \ge 0$. Thus,
\begin{align*}
&\int_{\mathfrak{T}_\alpha}^{K_\alpha(t)} \xi(K_\alpha(t) - s) R_0(u,\psi_1(s,u)) ds=  \\ &\int_{\mathfrak{T}_\alpha}^{t} \xi(K_\alpha(t) - K_\alpha(s)) R_0(u,\psi_1(K_\alpha(s),u)) \kappa_\alpha(s) ds \le \\
&\int_{\mathfrak{T}_\alpha}^{t} \xi(t - s) R_0(u,\psi_\alpha(s,u)) ds 
\end{align*}
where we have used \eqref{eq:K_concave} and the assumption that $\xi$ is increasing in the last inequality. Combining this estimate with 
\[\int_0^{\mathfrak{T}_\alpha} \xi(K_\alpha(t) - s) R_0(u,\psi_1(s,u)) \le \int_0^{\mathfrak{T}_\alpha} \xi(t - s) R_0(u,\psi_\alpha(s,u)) \]
part (a) follows. The proof of part (b) is analogous, replacing $K_\alpha$ by $L_{\alpha,\lambda}$ and using \eqref{eq:volterra_resolvent} as in the proof of Theorem~\ref{thm:mgf_comparison}.
\end{proof}

    %SECT: CRITICAL MOMENTS-----------------------------------------------------
\section{Comparison of critical moments}
For the rough Heston model $S$ with kernel $\kappa_\alpha$, $\alpha \in (1/2,1]$, the \textit{lower} resp.{} \textit{upper critical moments} are defined by 
\begin{subequations}\label{crit_mom_alpha}
\begin{align}
	  u_{\alpha}^-(t) &:= \inf\{u < 0: \EE[S_t^u] < \infty \}, \qquad t>0,\\
	  u_{\alpha}^+(t) &:= \sup\{u > 1: \EE[S_t^u] < \infty \}, \qquad t>0.
\end{align}
\end{subequations}
It is well-understood that these critical moments encode important information on the tail behavior of the marginal distributions of $S$. Moreover, the critical moments can be written in terms of moment explosion times as
\begin{subequations}\label{crit_mom_t}
\begin{align}
	  u_{\alpha}^-(t) &:= \inf\{u < 0: t < T_\alpha^*(u) \}, \qquad t>0,\\
	  u_{\alpha}^+(t) &:= \sup\{u > 1: t < T_\alpha^*(u) \}, \qquad t>0.
\end{align}
\end{subequations}
This suggests that under suitable conditions on $u \mapsto T_\alpha^*(u)$ the mappings $t \mapsto u_\alpha^\pm(t)$ are its piecewise inverse functions. In the case $\alpha = 1$ this is indeed the case, made precise in the following Lemma, which can be derived by elementary calculus from representation \eqref{eq:T_explicit} of $T_1^*(u)$:
\begin{lemma}\label{lem:T_inversion}
Let $\rho <0$ and let $d_\pm$ be defined as in \eqref{eq:dpm}. The function $u \mapsto T_*(u)$ is a strictly increasing continuous function from $(-\infty,d_-)$ onto $(0,\infty)$ and a strictly decreasing continuous function from $(d_+,\infty)$ onto $(0,\infty)$. Its inverse functions are given by $t \mapsto u_1^-(t)$ and $t \mapsto u_1^+(t)$ on the respective domains, and hence
\begin{equation}\label{eq: explosion_critmom_relation_1}
T_1^*(u_1^\pm(t)) = t, \qquad \forall t > 0.
\end{equation}
\end{lemma}
We remark that $d_\pm$ are precisely the boundaries between case (B) and (C) and that $T_*(u) = \infty$ for all $u \in [d_-,d_+]$. In the \textit{rough} Heston model ($\alpha < 1$) it is currently only known (from \cite{gerhold2018moment}) that $u \mapsto T_\alpha^*(u)$ are monotone functions (not necessarily in the strict sense) on the same domains as $T_1^*(u)$. For our purposes, however, the following property will be good enough: Directly from \eqref{crit_mom_t}, it follows that $u < u_\alpha^+(t)$ implies $t < T_\alpha^*(u)$ and $u > u_\alpha^-$ implies $t < T_\alpha^*(u)$. By contraposition, we obtain
\begin{equation}\label{eq:implication_T}
t \ge T_\alpha^*(u) \quad \Longrightarrow \quad \begin{cases} u \le u_\alpha^-(t) &\quad \text{if} \; u < d_-\\ u \ge u_\alpha^+(t) &\quad \text{if} \; u > d_+.\end{cases}
\end{equation}

In case (B), Theorem~\ref{thm:heston_upper_bound}, the key comparison principle for the moment explosion time $T_\alpha^*(u)$, is based on $\overline T_1^*(u)$ rather than on $T_1^*(u)$. Therefore, we also define
\begin{align}
	 \overline u_1^+(t) &:= \sup\{u > 1: t < \overline T_1^*(u)\}, \qquad t>0, \\ 
	\overline u_1^-(t) &:= \inf\{u < 0: t < \overline T_1^*(u)\}, \qquad t>0.
\end{align}
Note that there is no stochastic model for which $\overline u_1^\pm(t)$ represent the critical moments and therefore we refer to them as \textit{critical pseudo-moments}. In analogy to Lemma~\ref{lem:T_inversion}, the following can be derived by elementary calculus from \eqref{eq:Tbar}:
\begin{lemma}
Let $\rho < 0$, let $d_\pm$ be defined as in \eqref{eq:dpm} and set
\begin{equation}\label{eq:Tfrak}
\mathfrak{T}_\text{crit} := \overline{T}_1^*\left(\frac{\lambda}{\rho \eta}\right)  = \frac{|\rho| \pi}{\sqrt{\lambda(\lambda - \rho \eta)}}.
\end{equation}
The function $u \mapsto \overline T_*(u)$ is a strictly increasing continuous function from $(\lambda/(\rho \eta),d_-)$ onto $(\mathfrak{T}_\text{crit},\infty)$ and a strictly decreasing continuous function from $(d_+,\infty)$ onto $(0,\infty)$. Its inverse functions are given by $t \mapsto \overline u_1^-(t)$ and $t \mapsto \overline u_1^+(t)$ on the respective domains, and hence
\begin{equation}\label{eq:explosion_critmom_relation_2}
\begin{split}
\overline T_1^*(\overline u_1^-(t)) &= t, \qquad \forall t > \mathfrak{T}_\text{crit},\\
\overline T_1^*(\overline u_1^+(t)) &= t, \qquad \forall t > 0.
\end{split}
\end{equation}
\end{lemma}

We are now prepared to state our main comparison result on critical moments.

\begin{theorem}\label{thm:crit_mom_bound}Let $\rho <0$ and set 
\begin{equation}\label{eq:Tcrit_prime}
\mathfrak{T}'_\text{crit} := (\alpha \Gamma(\alpha) \mathfrak{T}_\text{crit} )^{1/\alpha} = \left(\frac{\alpha \Gamma(\alpha) |\rho| \pi}{\sqrt{\lambda(\lambda - \rho \eta)}}\right)^{1/\alpha}.
\end{equation}
Then the critical moments of the rough Heston model satisfy
\begin{subequations}\label{eq:u_estimates}
\begin{align}
u_\alpha^-(t) &\ge u_1^-\left(\frac{t^{\alpha}}{\alpha \Gamma(\alpha)} \right) \qquad \forall\, t \in (0, \mathfrak{T}'_\text{crit}] \label{eq:u_estimate_a}\\
u_\alpha^-(t) &\ge \overline u_1^-\left(\frac{t^{\alpha}}{\alpha \Gamma(\alpha)} \right) \qquad \forall\, t \in (\mathfrak{T}'_\text{crit},\infty)\\
u_\alpha^+(t) &\le \overline u_1^+\left(\frac{t^{\alpha}}{\alpha \Gamma(\alpha)} \right) \qquad \forall\, t \in (0,\infty).
\end{align}
\end{subequations}
For any $t \le \mathfrak{T}'_\alpha$ the inequalities also remain valid with $\tfrac{t^{\alpha}}{\alpha \Gamma(\alpha)}$ replaced by $t$.
 \end{theorem}
 \begin{proof}
First, observe that $u_1^-(t^\alpha/(\alpha \Gamma(\alpha)))$ is in the domain of case (A) if and only if 
\[\frac{t^\alpha}{\alpha \Gamma(\alpha)} \le T_1^*\left(\frac{\lambda}{\rho \eta}\right) = \mathfrak{T}_\text{crit},\]
which is easily transformed into
\[t \le \mathfrak{T}'_\text{crit} = \left(\frac{\alpha \Gamma(\alpha) |\rho| \pi}{\sqrt{\lambda(\lambda - \rho \eta)}}\right)^{1/\alpha}.\]
For any such $t$ we obtain, using Theorem~\ref{thm:heston_upper_bound} and \eqref{eq: explosion_critmom_relation_1}, that
 \begin{eqnarray*}
 	T_{\alpha}^*\left(u_1^-\left(\frac{t^{\alpha}}{\alpha \Gamma(\alpha)} \right)\right) &\overset{Thm.~\ref{thm:heston_upper_bound}}{\underset{\text{}}{\le}}&  \left(\alpha \Gamma(\alpha) T_1^*\left(u_1^-\left(\frac{t^{\alpha}}{\alpha \Gamma(\alpha)} \right)\right)\right)^{1/\alpha} \\
 	&\overset{\eqref{eq: explosion_critmom_relation_1}}{\underset{\text{}}{=}}& \left(\alpha \Gamma(\alpha) \left( \frac{t^{\alpha}}{\alpha \Gamma(\alpha)} \right) \right)^{1/ \alpha} \\
 	&=& t.
 \end{eqnarray*}
 By \eqref{eq:implication_T}, this implies $u \le u_\alpha^-(t)$, showing the first inequality of \eqref{eq:u_estimates}. The other two inequalities are shown analogously, but  -- owing to the fact that case (B) applies -- the critical pseudo-moments $\overline u_1^\pm(t)$ have to be used instead of $u_1^\pm(t)$. The last claim follows from the fact that $t \le t^\alpha/(\alpha \Gamma(\alpha))$ for all $t \le \mathfrak{T}_\alpha$ and the monotonicity of $u_1^\pm(.)$ and $\overline u_1^\pm(.)$.
 \end{proof}

	\section{Applications to Implied Volatility}\label{sec:implied}
	As known from the work of Roger Lee \cite{lee2004moment}, moment explosions and critical moments are closely related to the shape of the implied volatility smile for deep in-the-money or out-of-the-money options. In this section we will apply Lee's moment formula to our results and compare the smile's asymptotic steepness in the rough and classic Heston model.	
		
	For any given strike $K$ of a European option with maturity $T$, let $x = \log\left(\frac{K}{S_0}\right)$ denote the log-moneyness. Let $\sigma_{\text{iv}}(T,x)$ be the associated implied Black-Scholes volatility and define the \textit{asymptotic implied volatility slope} as 
	\begin{equation}\label{eq:AIVS}
	AIVS^{\pm}(T) = \limsup_{x \to \pm \infty} \sigma_{\text{iv}}^2(T,x)/|x|.
	\end{equation}
	Note that the superscript $\pm$ refers to the left ($-$) and right ($+$) wing of the smile respectively. We also remark that in most models of practical interest, such as the Heston model, the '$\limsup$' can be replaced by a genuine limit, e.g. by applying the theory of regularly varying functions; see \cite{benaim2009regular}. For the rough Heston model, however, it is currently an open question whether the $\limsup$ in \eqref{eq:AIVS} can be replaced by a genuine limit.\\
	The connection between critical moments and the asymptotic implied volatility slope is given by Lee's moment formula:
			
	\begin{proposition}[\cite{lee2004moment}] \label{prop:Lee}
	For all $T > 0$ it holds that
	\begin{equation*}
		AIVS^{-}(T) = \frac{\varsigma( - u^-(T) )}{T},
	\end{equation*}
	and
	\begin{equation*}
		AIVS^{+}(T) = \frac{\varsigma(  u^+(T)-1 )}{T},
	\end{equation*}
	where $ \varsigma(y) = 2 - 4(\sqrt{y^2+y} - y) $ and $u^{\pm}(T)$ are the critical moments.
\end{proposition}

Applying our comparison results for critical moments to Proposition \ref{prop:Lee} we obtain the following result.
\begin{theorem}\label{thm:implied_vola_bound}
	Let $\rho < 0$ and let $AIVS_{\alpha}^\pm (T)$ and $AIVS_{1}^\pm(T)$ be the asymptotic implied volatility slope in the rough resp. classic Heston model for maturity $T>0$. Then
	\begin{equation*}
		AIVS_{\alpha}^- (T) \ge \frac{T^{\alpha -1}}{\alpha \Gamma(\alpha )} AIVS_1^-\left(\frac{T^{\alpha}}{\alpha \Gamma(\alpha)}\right)
	\end{equation*}
	for all $T \le \mathfrak{T}'_\text{crit}\,$, with $\mathfrak{T}'_\text{crit}$ as in \eqref{eq:Tcrit_prime}.
\end{theorem}
\begin{remark}
This results shows that for small maturities the slope of left-wing implied volatility in the rough Heston model is dramatically steeper than in the non-rough Heston model. This complements known results on small-time behavior of the at-the-money skew in rough models, which explodes at the same rate, cf. \cite{fukasawa2017short}. 
\end{remark}
\begin{proof} Since $T \le \mathfrak{T}'_\text{crit}\,$, we can apply \eqref{eq:u_estimate_a} from Theorem~\ref{thm:crit_mom_bound} to estimate the lower critical moment. Since the function $\varsigma$ is strictly decreasing on $\RR_{ \ge 0}$, a straightforward application of Proposition~\ref{prop:Lee} yields
	\begin{align*}
		AIVS_{\alpha}^-(T) &= \frac{\varsigma \left(- u_{\alpha}^-(T) \right)}{T} \\
		&\ge \frac{\varsigma \left(- u_{1}^-(\frac{T^{\alpha}}{\alpha \Gamma(\alpha)}) \right)}{T} \\
		&= \frac{T^{\alpha -1}}{\alpha \Gamma(\alpha )} \frac{\varsigma \left(- u_{1}^-(\frac{T^{\alpha}}{\alpha \Gamma(\alpha)}) \right)}{\frac{T^{\alpha}}{\alpha \Gamma(\alpha)}} \\
		&= \frac{T^{\alpha -1}}{\alpha \Gamma(\alpha )} AIVS_1^-\left(\frac{T^{\alpha}}{\alpha \Gamma(\alpha)} \right). \qedhere
	\end{align*}
\end{proof}
Theorem~\ref{thm:implied_vola_bound}, which is non-asymptotic in $T$, can be complemented by another result, which is asymptotic in $T$, but also contains information on the right-wing implied volatility slope. Here and below, we use the notation
\begin{align*}
f(t) \sim g(t) \quad &\Longleftrightarrow \quad \lim_{t \to 0} \frac{f(t)}{g(t)} = 1,\\
f(t) \gtrsim g(t) \quad &\Longleftrightarrow \quad \lim_{t \to 0} \frac{f(t)}{g(t)} \ge 1,
\end{align*}
and apply it also to other limits (e.g. $t \to \infty$) when indicated.

\begin{theorem}\label{thm:AIVS_smallT}
Let $\rho <0$ and set
\[C_\pm = \pi  - 2\arctan\left(\frac{\pm \rho}{\sqrt{1 - \rho^2}}\right), \qquad  D = \pi - \frac{2\rho}{\sqrt{1 - \rho^2}}.\]
In the classic Heston model, the limits $AIVS_1^\pm(0) := \lim_{T \downarrow 0} AIVS_1^\pm(T)$ exist and are given by
\begin{equation}
AIVS_1^\pm(0) = \frac{\eta \sqrt{1 - \rho^2}}{2 C_\pm}. 
\end{equation}
In the rough Heston model, it holds that
\begin{subequations}\label{eq:AIVS_small_t}
\begin{align}
AIVS_\alpha^-(T) &\gtrsim \frac{T^{\alpha -1}}{\alpha \Gamma(\alpha )} AIVS_1^-(0) \qquad (\text{as }T \to 0) \label{eq:AIVS_small_t-}\\
AIVS_\alpha^+(T) &\gtrsim \frac{C_+}{D} \frac{T^{\alpha -1}}{\alpha \Gamma(\alpha )} AIVS_1^+(0) \qquad (\text{as }T \to 0). \label{eq:AIVS_small_t+}
\end{align}
\end{subequations}
\end{theorem}
\begin{remark}
This result shows that as $T \to 0$ the right-wing asymptotic implied volatility slope of the rough Heston model explodes at the same power-law rate as the left-wing asymptotic implied volatility slope. We remark that the constant $C_+/D$ which distinguishes the estimates at the left and the right wing, is always within $(0,1)$, given that $\rho < 0$.
\end{remark}
\begin{proof}We first analyze the behavior of $T_1^*(u)$ and  $\overline T_1^*(u)$ as $|u| \to \infty$. To this end note that it follows from \eqref{eq:discriminant} that $\Delta(u) <  0$ for $|u|$ large enough and that 
\[\lim_{u \to \pm \infty} \frac{c_2(u)}{2 \sqrt{-\Delta(u)}} = \frac{\pm \rho}{\sqrt{1 - \rho^2}}.\]
Inserting into \eqref{eq:T_explicit}, we obtain
\[T_1^*(u) \sim |u|^{-1}\frac{C_\pm}{\eta \sqrt{1 - \rho^2}}  \qquad (\text{as }u \to \pm \infty),\]
and from \eqref{eq:Tbar}, we obtain
\[\overline T_1^*(u) \sim u^{-1}\frac{D}{\eta \sqrt{1 - \rho^2}} \qquad (\text{as }u \to \infty).\]
The critical (pseudo-)moments are the piecewise inverse functions of the moment explosion times, and hence
\begin{align*}
u_1^\pm(t) &\sim -t^{-1} \frac{C_\pm}{\eta \sqrt{1 - \rho^2}} \qquad (\text{as }t \to 0)\\
\overline u_1^+(t) &\sim t^{-1} \frac{D}{\eta \sqrt{1 - \rho^2}} \qquad (\text{as }t \to 0).
\end{align*}
To obtain the small-time behaviour of the asymptotic implied volatility slope, it remains to insert these relations into Lee's moment formula. First, note that
\[\lim_{\epsilon \to 0} \frac{\varsigma(1/\epsilon)}{\epsilon}  = \frac{1}{2}.\]
Hence, with focus on the right wing, we conclude that for the classic Heston model
\begin{equation}\label{eq:AIVS_interm}
\lim_{T \to 0} AIVS_1^+(T) = \lim_{T \to 0} \frac{\varsigma(u_1^+(T) - 1)}{T} = \frac{\eta \sqrt{1 - \rho^2}}{2C_+},
\end{equation}
and similarly at the left wing. For the rough Heston model we estimate with Theorem~\ref{thm:crit_mom_bound} (again at the right wing) 
\begin{align*}
\lim_{T \to 0} \frac{\alpha \Gamma(\alpha)}{T^{\alpha-1}} AIVS_\alpha^+(T) &= \lim_{T \to 0}  \frac{\alpha \Gamma(\alpha)}{T^{\alpha}} \varsigma(u_\alpha^+(T) - 1) \\
&\ge \lim_{T \to 0}  \frac{\alpha \Gamma(\alpha)}{T^{\alpha}} \varsigma\left(\overline u_1^+\left(\frac{T^\alpha}{\alpha \Gamma(\alpha)}\right) - 1\right) \\
&= \lim_{T \to 0} \alpha \Gamma(\alpha) \left\{T^\alpha \left(\overline u_1^+\left(\frac{T^\alpha}{\alpha \Gamma(\alpha)}\right) - 1\right) \right\}^{-1} \cdot 
\frac{1}{2} \\
&= \frac{\eta \sqrt{1 - \rho^2}}{2 D}.
\end{align*}
Comparison with \eqref{eq:AIVS_interm} yields \eqref{eq:AIVS_small_t+}. The calculation on the left wing uses $-u_1^-$ instead of $\overline u_1^+ - 1$ and gives
\[\lim_{T \to 0} \frac{\alpha \Gamma(\alpha)}{T^{\alpha-1}} AIVS_\alpha^- (T) \ge  \frac{\eta \sqrt{1 - \rho^2}}{2 C_-},\]
completing the proof.
\end{proof}

\section{Numerical Illustration}
In this section we graphically illustrate and compare the bounds of moment explosions (Thm.~\ref{thm:heston_upper_bound}) and of the asymptotic implied volatility slope (Thms.~\ref{thm:implied_vola_bound} and \ref{thm:AIVS_smallT}) for a concrete choice of the rough Heston model's parameters. We set
\begin{align*}
	\rho &= -0.8, \\
	\lambda &= 2, \\
	\eta &= 0.2,
\end{align*}
and
\begin{equation*}
	\alpha = 0.6,
\end{equation*}
which corresponds to a Hurst parameter of $H=0.1$, which is close to the volatility-based estimate of \cite{gatheral2018volatility}.

\subsection{Moment explosion times}
To provide a better readability, we use the following notations:
\begin{itemize}
	\item  $T_{\text{KM}}^+(u) = \begin{cases} 
 \left( \alpha \Gamma( \alpha) T_1^*(u) \right)^{1/ \alpha} & u \text{ sat. case (A)} \\
\left( \alpha \Gamma( \alpha) \overline T_1^*(u) \right)^{1/ \alpha} & u \text{ sat. case (B)}
\end{cases},
$  \\denotes the combined upper bounds of $T_{\alpha}^*(u)$, introduced in Theorem~\ref{thm:heston_upper_bound},
	\item  $T_{\text{GGP}}^+(u)$ and $T_{\text{GGP}}^-(u)$ denote the upper and lower bound of $T_{\alpha}^*(u)$, introduced in Theorem 4.1 and 4.2 in \cite{gerhold2018moment},
	\item  $T^*_{\alpha,\text{aprx}}(u)$ denotes the approximation of the explosion time $T_{\alpha}^*(u)$, computed by Algorithm 7.5 in \cite{gerhold2018moment}, which is valid for $u$ in case (A).
\end{itemize}

Figure \ref{fig:bounds_alpha0.6} shows a comparison of the moment explosion bounds and the approximation $T^*_{\alpha,\text{aprx}}(u)$ in the given setting. It can be seen that the bound $T_{\text{KM}}^+(u)$ is tighter than $T_{\text{GGP}}^+(u)$ on both sides of $u=0$. Numerical experiments confirm that this relation persists in a large range of parameters, except in very close proximity to the boundary case $\alpha = 0.5$.

%%%%%%%%%%%%%%%%%%%%%%%%%%%%%%%%%%%%%%%%%%%%%%%
% plot: bounds of explosion time, alpha = 0.6
%%%%%%%%%%%%%%%%%%%%%%%%%%%%%%%%%%%%%%%%%%%%%%%
\begin{figure}[htbp]
	\includegraphics[width=0.75\textwidth]{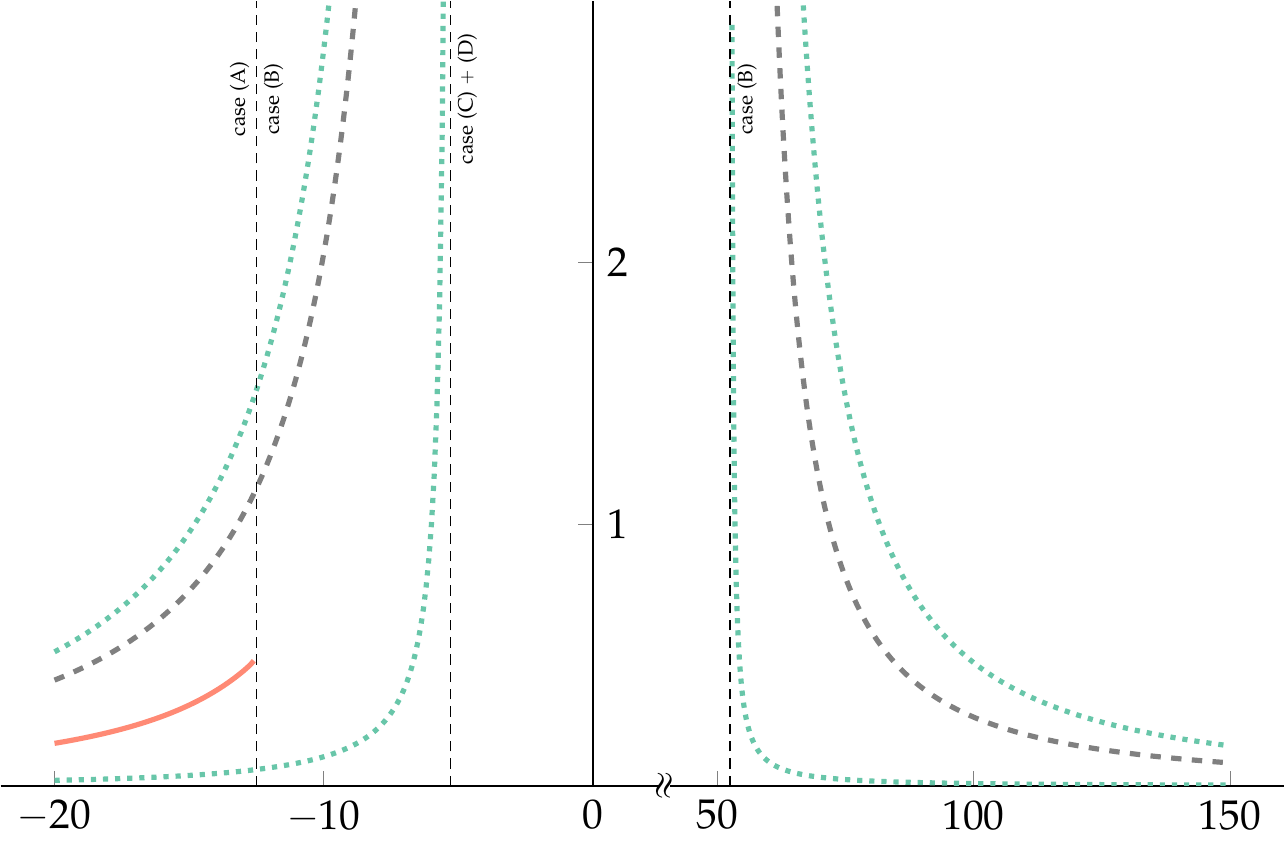}
		\caption{Bounds of the moment explosion time in the rough Heston model for $u \in [-20,150]$ and $\alpha = 0.6$. Grey dashed curves: Our combined upper bound $T_{\text{KM}}^+(u)$, green dotted curves: $T_{\text{GGP}}^+(u)$ and $T_{\text{GGP}}^-(u)$ from \cite{gerhold2018moment}, red solid curve: approximation $T^*_{\alpha,\text{aprx}}(u)$ of true explosion time, valid for $u \le \tfrac{\lambda}{\rho \eta} = -12.5$.}\label{fig:bounds_alpha0.6}
\end{figure}

\subsection{Implied volatility asymptotics} 
In Figures~\ref{fig:vola_neg} and \ref{fig:vola_pos} we illustrate the bounds for the asymptotic implied volatility slope from Theorems~\ref{thm:implied_vola_bound} and \ref{thm:AIVS_smallT}. The bounds shown in the plots are generated as follows: First, we use $T_1^*(u)$ as function of $u$ to compute the critical moments $u^{\pm}_{1}(t)$ of the classic Heston model by numerical root finding. Afterwards we use Lee's moment formula to determine $AIVS^\pm_1(T)$, the asymptotic implied volatility slope in the classic Heston model. The bounds of the rough Heston implied volatility slope $AIVS^\pm_\alpha$ from Theorems~\ref{thm:implied_vola_bound} and \ref{thm:AIVS_smallT} are then computed from $AIVS^\pm_1$. On the left wing of the smile (where case (A) applies for small $T$), we additionally compute an approximation of $AIVS^-_\alpha(T)$, by applying the same procedure to the approximate explosion time $T^*_{\alpha,\text{aprx}}(u)$. 
	
	\begin{figure}[htbp]
\includegraphics[width=0.75\textwidth]{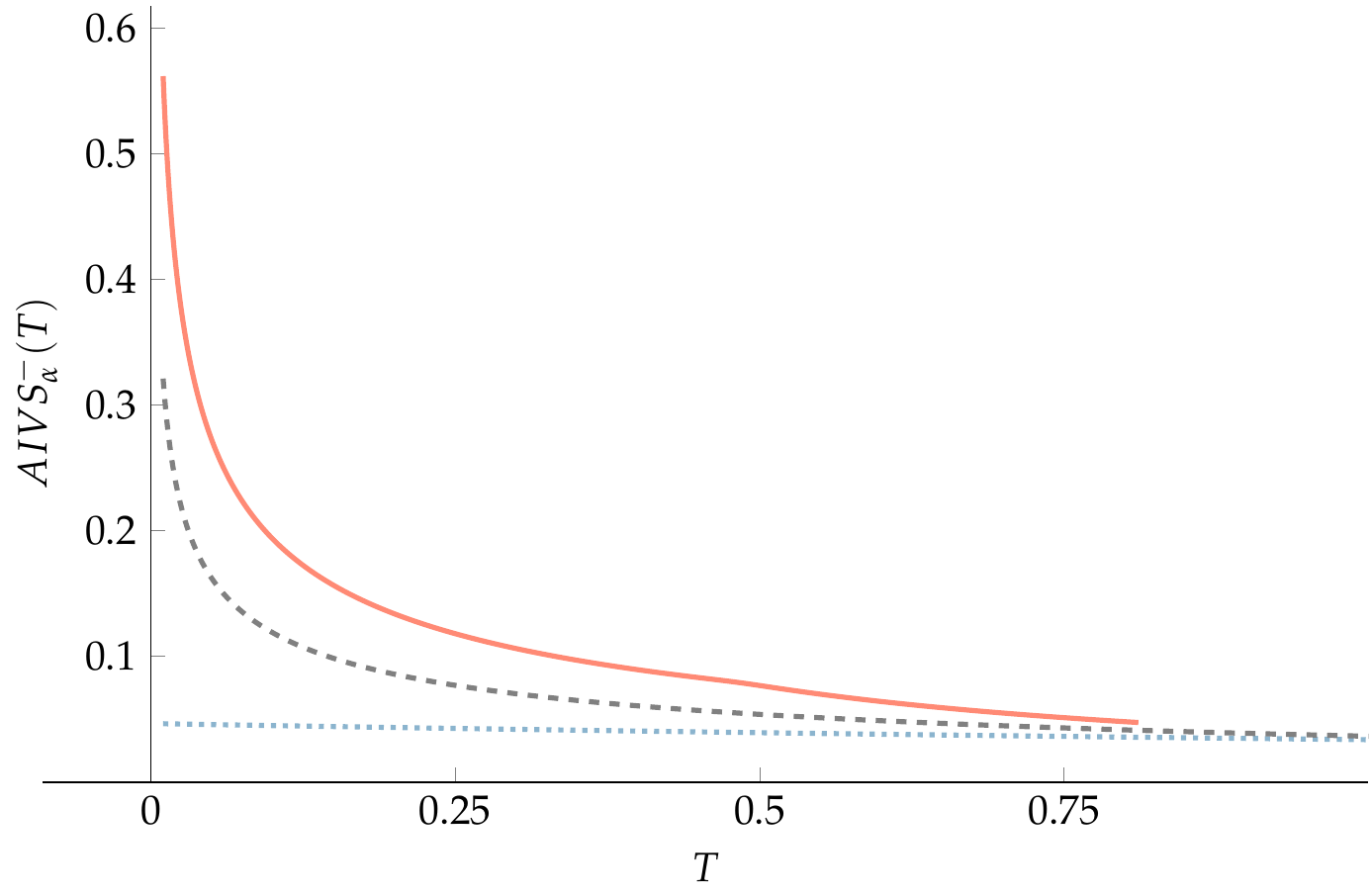}
	\caption{Left-wing asymptotic implied volatility slope ($AIVS^-$) in the classic Heston model (blue dotted), the lower bound from Theorem~\ref{thm:implied_vola_bound} for the rough Heston model (grey dashed) with $\alpha = 0.6$, and an approximation of $AIVS_{\alpha}^-(T)$, derived from $T^*_{\alpha,\text{aprx}}(u)$ (red solid),  valid in case (A), i.e., up to $T^*_{\alpha,\text{aprx}}(\lambda/(\rho \eta)) \approx 0.81$.}\label{fig:vola_neg}
\end{figure}

\begin{figure}[htbp]
	\includegraphics[width=0.75\textwidth]{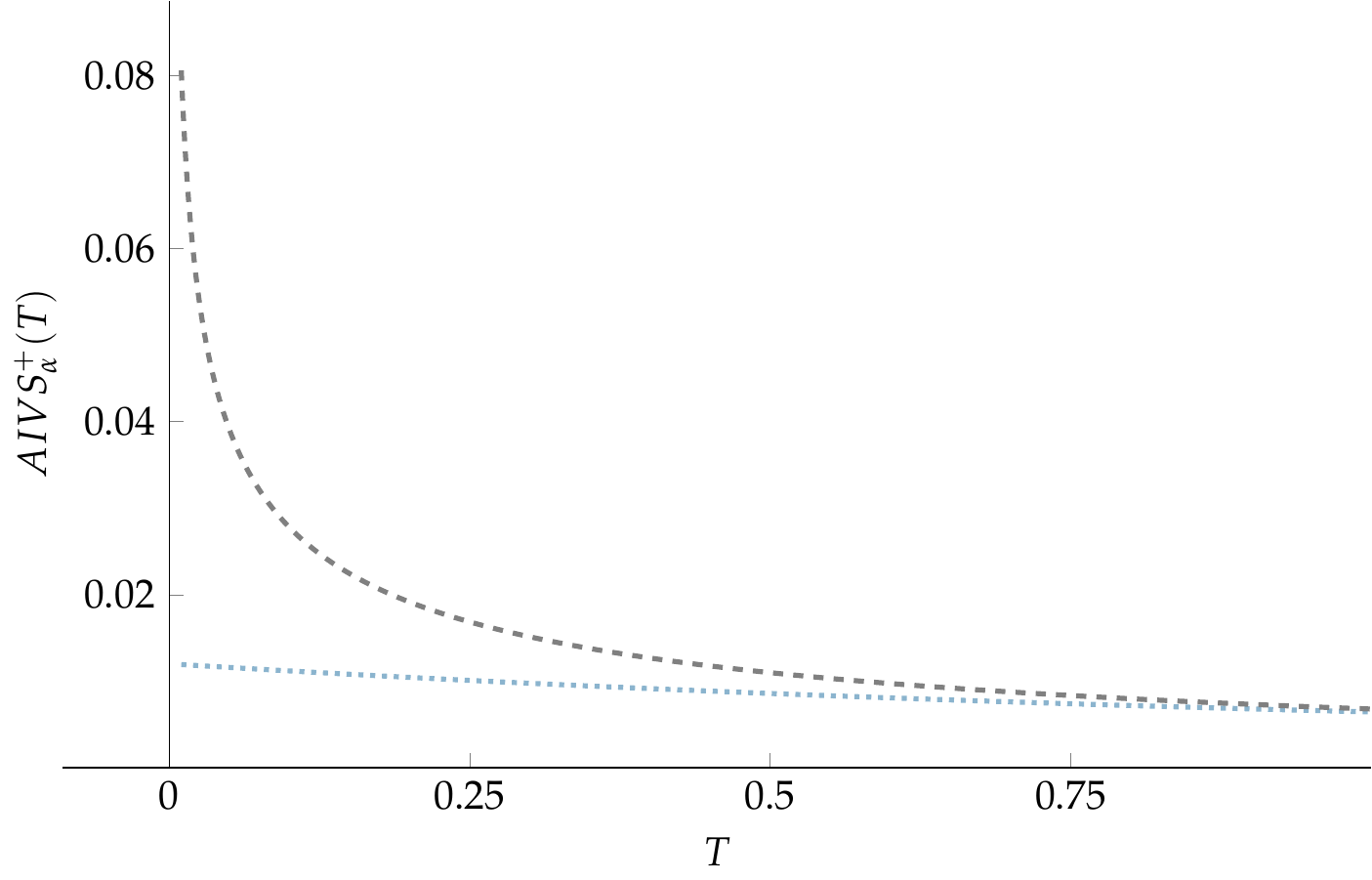}
		\caption{Right-wing asymptotic implied volatility slope ($AIVS^+$) in the classic Heston model (blue dotted), and asymptotic lower bound \eqref{eq:AIVS_small_t+} from Theorem~\ref{thm:AIVS_smallT} for the rough Heston model (grey dashed) with $\alpha = 0.6$.}\label{fig:vola_pos}
	\end{figure}

\bibliographystyle{alpha}
\bibliography{refs}

\newcommand{\etalchar}[1]{$^{#1}$}
\begin{thebibliography}{JKRM13}

\bibitem[AJLP17]{jaber2017affine}
Eduardo Abi~Jaber, Martin Larsson, and Sergio Pulido.
\newblock Affine {V}olterra processes.
\newblock {\em arXiv:1708.08796}, 2017.

\bibitem[AP07]{andersen2007moment}
Leif B.~G. Andersen and Vladimir~V. Piterbarg.
\newblock Moment explosions in stochastic volatility models.
\newblock {\em Finance and Stochastics}, 11(1):29--50, Jan 2007.

\bibitem[Bat96]{bates1996jumps}
David~S. Bates.
\newblock Jumps and stochastic volatility: Exchange rate processes implicit in
  deutsche mark options.
\newblock {\em The Review of Financial Studies}, 9(1):69--107, 1996.

\bibitem[BCC97]{bakshi1997empirical}
Gurdip Bakshi, Charles Cao, and Zhiwu Chen.
\newblock Empirical performance of alternative option pricing models.
\newblock {\em Journal of Finance}, 52(5):2003--49, 1997.

\bibitem[BF09]{benaim2009regular}
Shalom Benaim and Peter Friz.
\newblock Regular variation and smile asymptotics.
\newblock {\em Mathematical Finance: an International Journal of Mathematics,
  Statistics and Financial Economics}, 19(1):1--12, 2009.

\bibitem[BFG{\etalchar{+}}19]{bayer2019short}
Christian Bayer, Peter~K Friz, Archil Gulisashvili, Blanka Horvath, and
  Benjamin Stemper.
\newblock Short-time near-the-money skew in rough fractional volatility models.
\newblock {\em Quantitative Finance}, 19(5):779--798, 2019.

\bibitem[BFL09]{benaim2009black}
Shalom Benaim, Peter Friz, and Roger Lee.
\newblock On black-scholes implied volatility at extreme strikes.
\newblock {\em Frontiers in Quantitative Finance: Volatility and Credit Risk
  Modeling}, pages 19--45, 2009.

\bibitem[B{\"u}h06]{buehler2006volatility}
Hans B{\"u}hler.
\newblock {\em Volatility Markets -- Consistent modeling, hedging and practical
  implementation}.
\newblock PhD thesis, TU Berlin, 2006.

\bibitem[BW98]{browder1998ordinary}
Andrew Browder and Wolfgang Walter.
\newblock {\em Ordinary Differential Equations}.
\newblock Graduate Texts in Mathematics. Springer New York, 1998.

\bibitem[ER18]{eleuch2018perfect}
Omar {El Euch} and Mathieu Rosenbaum.
\newblock Perfect hedging in rough heston models.
\newblock {\em The Annals of Applied Probability}, 28(6):3813--3856, 2018.

\bibitem[ER19]{eleuch2019characteristic}
Omar {El Euch} and Mathieu Rosenbaum.
\newblock The characteristic function of rough heston models.
\newblock {\em Mathematical Finance}, 29(1):3--38, 2019.

\bibitem[FGS19]{forde2019small}
Martin Forde, Stefan Gerhold, and Benjamin Smith.
\newblock Small-time and large-time smile behaviour for the rough heston model.
\newblock {\em preprint}, 2019.

\bibitem[FKR10]{friz2010encyclopedia}
Peter Friz and Martin Keller-Ressel.
\newblock Moment explosions.
\newblock In Rama Cont, editor, {\em Encyclopedia of Quantitative Finance}.
  Wiley, 2010.

\bibitem[FTW19]{fukasawa2019volatility}
Masaaki Fukasawa, Tetsuya Takabatake, and Rebecca Westphal.
\newblock Is volatility rough?
\newblock {\em arXiv:1905.04852}, 2019.

\bibitem[Fuk17]{fukasawa2017short}
Masaaki Fukasawa.
\newblock Short-time at-the-money skew and rough fractional volatility.
\newblock {\em Quantitative Finance}, 17(2):189--198, 2017.

\bibitem[FZ17]{forde2017asymptotics}
Martin Forde and Hongzhong Zhang.
\newblock Asymptotics for rough stochastic volatility models.
\newblock {\em SIAM Journal on Financial Mathematics}, 8(1):114--145, 2017.

\bibitem[Gat06]{gatheral2006volatility}
Jim Gatheral.
\newblock {\em The Volatility Surface: A Practitioner's Guide}.
\newblock Wiley Finance. Wiley, 2006.

\bibitem[GGP18]{gerhold2018moment}
Stefan Gerhold, Christoph Gerstenecker, and Arpad Pinter.
\newblock Moment explosions in the rough heston model.
\newblock arXiv:1801.09458, 2018.

\bibitem[GJR18]{gatheral2018volatility}
Jim Gatheral, Thibault Jaisson, and Mathieu Rosenbaum.
\newblock Volatility is rough.
\newblock {\em Quantitative Finance}, 18(6):933--949, 2018.

\bibitem[GJRS18]{guennoun2018asymptotic}
Hamza Guennoun, Antoine Jacquier, Patrick Roome, and Fangwei Shi.
\newblock Asymptotic behavior of the fractional heston model.
\newblock {\em SIAM Journal on Financial Mathematics}, 9(3):1017--1045, 2018.

\bibitem[GKR19]{gatheral2018affine}
Jim Gatheral and Martin Keller-Ressel.
\newblock Affine forward variance models.
\newblock {\em Finance \& Stochastics}, 2019.
\newblock doi: 10.1007/s00780-019-00392-5.

\bibitem[GLS90]{gripenberg1990volterra}
G.~Gripenberg, S.~O. Londen, and O.~Staffans.
\newblock {\em Volterra Integral and Functional Equations}.
\newblock Encyclopedia of Mathematics and its Applications. Cambridge
  University Press, 1990.

\bibitem[Hes93]{heston1993closed}
Steven~L. Heston.
\newblock A closed-form solution for options with stochastic volatility with
  applications to bond and currency options.
\newblock {\em Review of Financial Studies}, 6:327--343, 1993.

\bibitem[HMS11]{haubold2011mittag}
Hans~J Haubold, Arak~M Mathai, and Ram~K Saxena.
\newblock {M}ittag-{L}effler functions and their applications.
\newblock {\em Journal of Applied Mathematics}, 2011.

\bibitem[Hoh98]{hoh1998symbolic}
Walter Hoh.
\newblock A symbolic calculus for pseudo-differential operators generating
  feller semigroups.
\newblock {\em Osaka journal of mathematics}, 35(4):789--820, 1998.

\bibitem[JKRM13]{jacquier2013large}
Antoine Jacquier, Martin Keller-Ressel, and Aleksandar Mijatović.
\newblock Large deviations and stochastic volatility with jumps: asymptotic
  implied volatility for affine models.
\newblock {\em Stochastics}, 85(2):321--345, 2013.

\bibitem[KR11]{keller-ressel2011moment}
Martin Keller-Ressel.
\newblock Moment explosions and long-term behavior of affine stochastic
  volatility models.
\newblock {\em Mathematical Finance}, 21(1):73--98, 2011.

\bibitem[KRLP18]{keller-ressel2018affine}
Martin Keller-Ressel, Martin Larsson, and Sergio Pulido.
\newblock Affine rough models.
\newblock {\em arXiv:1812.08486}, 2018.

\bibitem[KST06]{kilbas2006theory}
Anatolii~Aleksandrovich Kilbas, Hari~Mohan Srivastava, and Juan~J. Trujillo.
\newblock {\em Theory and Applications of Fractional Differential Equations,
  Volume 204 (North-Holland Mathematics Studies)}.
\newblock Elsevier Science, New York, USA, 2006.

\bibitem[Lee04]{lee2004moment}
Roger~W. Lee.
\newblock The moment formula for implied volatility at extreme strikes.
\newblock {\em Mathematical Finance}, 14(3):469--480, 2004.

\end{thebibliography}
\end{document}